\documentclass[12pt,a4paper]{article}

\usepackage[english]{babel}
\usepackage[a4paper,top=2cm,bottom=2.54cm,left=1.91cm,right=1.91cm]{geometry}
\usepackage[numbers,sort&compress,square]{natbib}
\usepackage{amsmath}
\usepackage{amssymb}
\usepackage{amsthm}
\usepackage{graphicx}
\usepackage{amsfonts}
\usepackage{caption}
\usepackage{authblk}
\usepackage{float}
\usepackage{setspace}
\usepackage{titlesec}
\titleformat{\section}{\normalfont\Large\bfseries}{\thesection.}{1em}{}
\titleformat{\subsection}{\normalfont\large\bfseries}{\thesubsection.}{1em}{}
\usepackage[colorlinks=true, allcolors=blue]{hyperref}

\theoremstyle{plain}
\newtheorem{theorem}{Theorem}
\newtheorem{proposition}[theorem]{Proposition}
\newtheorem{corollary}[theorem]{Corollary}

\title{Generalized Boltzmann-Gibbs Distributions and the Electronic Partition Function Paradox}
\author[1,2]{Leandro Lyra Braga Dognini\thanks{E-mail: \href{mailto:leandro.dognini@uerj.br}{leandro.dognini@uerj.br}. All figures were drawn and the revision of this paper was done with the assistance of Claude, Anthropic's large language model.}}
\affil[1]{\small Rio de Janeiro State University, Rua São Francisco Xavier 524, 20550-900, Rio de Janeiro, Brazil}
\affil[2]{\small Legislative Advisory, Federal Senate of Brazil, Praça dos Três Poderes, 70165-900, Brasília, Brazil}
\date{November 16, 2026}

\begin{document}
\maketitle
\begin{abstract}
\noindent The maximization of the nonadditive entropic functional $S_{q}(p)=k\sum_{i}p_{i}\ln_{q}(1/p_{i})$, $q\in(0,1)$, under a linear mean-energy constraint, i.e., the best-behaved case from Dognini and Tsallis (2026), is used to derive generalizations of the Boltzmann-Gibbs distribution for physical systems whose description requires a limiting process over the energy spectrum: the hydrogen atom in free space; an energy spectrum that uniformly approaches a continuous, unbounded limit with a common degeneracy; the Maxwell-Boltzmann distribution of velocities; the harmonic oscillator; and the one-dimensional box. For the hydrogen atom with mean energy $U\in[0,e_{ion}]$, the limiting optimal distribution is a two-tier one, concentrated on the ground level and on the ionization threshold, with weights $1-U/e_{ion}$ and $U/e_{ion}$ for every $q\in(0,1)$, and the temperature rescaled by $\ln_{q}W$ is finite and given by $T_{q,s}=T_{ion}(U/e_{ion})^{1-q}$, with $T_{ion}=e_{ion}/kq$. This circumvents the Electronic Partition Function Paradox and yields a family of well-defined thermodynamic behaviors.
\end{abstract}

\textbf{Keywords}: nonadditive entropies, hydrogen atom, nonextensive statistical mechanics.

\section{Introduction}\label{sec1}

The canonical distribution is obtained through the maximization of the Boltzmann-Gibbs-von Neumann-Shannon entropic functional \cite{Boltzmann1872,Boltzmann1877,Gibbs1902,vonNeumann1927,Shannon1948} as
\begin{eqnarray}\label{maxProblemBG}
    \max & S_{BG}(p)\\
\text{s.t.}& p_{i}\geq0, 1\leq i\leq W\nonumber\\
&\sum^{W}_{i=1}p_{i}=1\nonumber\\
&\sum^{W}_{i=1}p_{i}e_{i}=U,\nonumber
\end{eqnarray}
with the entropic functional\footnote{I write $\mathbb{R}^W_{+}\equiv\{x\in\mathbb{R}^{W}\mid x_{i}\geq0,1\leq i\leq W\}$ and $\mathbb{R}^W_{++}\equiv\{x\in\mathbb{R}^{W}\mid x_{i}>0,1\leq i\leq W\}$.} $S_{BG}:\mathbb{R}^{W}_{+}\rightarrow\mathbb{R}$ given by
\begin{eqnarray*}
    S_{BG}(p)=k\sum^{W}_{i=1}p_{i}\ln \frac{1}{p_{i}},
\end{eqnarray*}
$k>0$ being the Boltzmann constant, $W\geq 2$ the number of possible states of the system, $e_{1}\leq\ldots\leq e_{W}$ the energy of each state (all such values form the \textit{energy spectrum}), and $e_{1}\le U \le e_{W}$ the mean energy of the system. I write $E_{i}=e_{i}-U$, $1\leq i\leq W$, for the \textit{relative energy spectrum} and $\Omega(e)=\#\{1\leq i\leq W\mid e_{i}=e\}$ for the \textit{degeneracy function}.

It is well known that, for $U\in(e_{1},e_{W})$, (\ref{maxProblemBG}) has a unique solution\footnote{The tilde always refers to variables related to an optimal solution of the respective optimization problem.} $\tilde{p}\in\mathbb{R}^{W}_{++}$ given by
\begin{eqnarray}\label{eqBGDistribution}
    \tilde{p}_{i}=\frac{\exp(-\beta_{BG} E_{i})}{\sum^{W}_{j=1}\exp(-\beta_{BG} E_{j})}=\frac{\exp(-\beta_{BG} e_{i})}{\sum^{W}_{j=1}\exp(-\beta_{BG} e_{j})},
\end{eqnarray}
for $1\leq i\leq W$, with the \textit{structural parameter} $\beta_{BG}\in\mathbb{R}$ satisfying
\begin{eqnarray}\label{eqBGBeta}
    \sum^{W}_{i=1}\exp(-\beta_{BG} E_{i})E_{i}=0.
\end{eqnarray}
Following \cite{DogniniTsallis2026}, $\beta_{BG}$ is called ``structural'' because it depends solely on the relative energy spectrum. The temperature is given by the Clausius relation \cite{Clausius1865} $\partial S_{BG}/\partial U=1/T_{BG}$, which yields $T_{BG}=1/k\beta_{BG}$.

Many physical systems of interest have an infinite energy spectrum and must be modeled through a limiting process, i.e., by solving (\ref{maxProblemBG}) for a truncated spectrum and taking the limit of the solution as the truncation is removed. The paradigmatic example is the hydrogen atom in free space, whose energy spectrum without degeneracy is
\begin{eqnarray}\label{eqHSpectrum}
    e_{n}=\biggl(1-\frac{1}{n^{2}}\biggr) e_{ion},
\end{eqnarray}
for $n\geq1$, with $e_{ion}=13.6\text{ eV}$ being the ionization energy and $\Omega(e_{n})=2n^{2}$. Since $e_{n}\rightarrow e_{ion}$ and $\Omega(e_{n})\rightarrow+\infty$, we have, for every $T_{BG}>0$,
\begin{eqnarray*}
\sum^{\infty}_{n=1}\Omega(e_{n})\exp\biggl(-\frac{e_{n}}{kT_{BG}}\biggr)=+\infty,
\end{eqnarray*}
and this divergence of the partition function characterizes what has been called in the literature the \textit{Electronic Partition Function Paradox} \cite{Strickler1966}. The paradox can also be seen from a different angle: if one solves (\ref{eqBGBeta}) for the spectrum truncated at its $N$-th distinct energy level, thus obtaining $T^{N}_{BG}=1/k\beta^{N}_{BG}$, then $\lim_{N\rightarrow\infty}T^{N}_{BG}=0$ for every mean energy $U\in(0,e_{ion})$ (see Proposition \ref{propParadox} below). In other words, within the Boltzmann-Gibbs framework, a hydrogen atom in free space with any positive mean energy is at zero temperature. Rescaling the temperature does not solve the problem either: rescaled by the logarithm of the number of states, $T^{N}_{BG}$ converges to $e_{ion}/k$ for every $U\in(0,e_{ion})$, so that the temperature does not determine the mean energy.

The literature furnishes different methods for dealing with this divergence. For instance, one may limit the energy spectrum (e.g., to the levels whose orbits do not exceed the mean interparticle distance \cite{Herzfeld1916,Fermi1924}), possibly grouping the remaining levels into a few tiers \cite{DAmmandoColonnaPietanzaCapitelli2010}, use zeta function regularization \cite{PlastinoRoccaFerri2019}, or generalize (\ref{maxProblemBG}) through nonadditive entropic functionals \cite{LucenaSilvaTsallis1995}. In the present paper, I follow the last route. I replace $S_{BG}$ by the nonadditive entropic functional $S_{q}$ \cite{Tsallis1988,Tsallis2023}, $q\in(0,1)$, keeping the \emph{linear} mean-energy constraint of (\ref{maxProblemBG}).

This is precisely the best-behaved case identified in \cite{DogniniTsallis2026} when analyzing the influence of constraints upon the optimization of $S_{q}$: a unique and strictly positive solution always exists, it is of ${(2-q)}$-exponential form, and the Third Law of Thermodynamics is preserved. 

I show that this framework yields a two-tier analytical model for the hydrogen atom in free space, with a finite and smooth temperature for every $U\in(0,e_{ion})$ and every $q\in(0,1)$, thus circumventing the paradox. I also apply this framework to four systems that likewise require a limiting process (an energy spectrum that uniformly approaches a continuous, unbounded limit with a common degeneracy; the Maxwell-Boltzmann distribution of velocities; the harmonic oscillator; and the one-dimensional box), recovering the Boltzmann-Gibbs results as $q\rightarrow1$.

After this \hyperref[sec1]{Introduction}, the outline of the paper is as follows. Subsection \ref{secLit} reviews the literature. Section \ref{sec2} recalls the results of \cite{DogniniTsallis2026} for the linear constraint case with $q\in(0,1)$ and establishes a general result on the limit of truncated energy spectra (Theorem \ref{theoTrunc}), both of which are used throughout the paper. Section \ref{secGen} presents the generalized Boltzmann-Gibbs distributions for the hydrogen atom in free space (including the resolution of the Electronic Partition Function Paradox); an energy spectrum that uniformly approaches a continuous, unbounded limit; the Maxwell-Boltzmann distribution of velocities; the harmonic oscillator; and the one-dimensional box. Section \ref{sec7} presents the concluding remarks, and the proofs are given in the \hyperref[app]{Appendix}.

\subsection{Related literature}\label{secLit}

The divergence of the partition function of the hydrogen atom was recognized in the early days of quantum theory \cite{Herzfeld1916,Fermi1924}, and its discussion under the name \textit{Electronic Partition Function Paradox} goes back to \citet{Strickler1966}. Within the Boltzmann-Gibbs framework, the direct way out is to truncate the spectrum (i.e., for some $N\geq1$, discard all energy levels $e_{n}$, $n>N$) on physical grounds: the criterion recalled by \citet{Fermi1924}, who attributes it to \citet{Herzfeld1916}, discards the levels whose orbits exceed the mean interparticle distance, and \citet{Miranda2001} argues that highly excited orbits exceed the size of any container. In a related approach, \citet{DAmmandoColonnaPietanzaCapitelli2010} approximate the truncated sum by combining the remaining energy levels into two or three grouped levels (i.e., a two- or three-tier model).

Other approaches act on the sum itself. \citet{Fermi1924} derives, from a thermodynamic argument that takes the atomic volume into account, a factor that multiplies each term of the sum, decreases exponentially with the volume of the corresponding orbit and thus ensures convergence. Similarly, \citet{HummerMihalas1988} weight each level by the probability that it survives the perturbations of the surrounding plasma (i.e., they also multiply each term of the sum by a factor that ensures convergence) and \citet{Blinder1995} derives the partition function from the Coulomb propagator. Also, \citet{OliveiraNetoCuradoNobreRegoMonteiro2007} regularize the divergent sum (i.e., make it converge) by subtracting from each term its limiting value (see Eqs.~(10) and (11) therein), and \citet{PlastinoRoccaFerri2019} perform such regularization through Riemann's zeta function. In all these approaches, the Boltzmann-Gibbs entropy is kept.

A different route replaces the Boltzmann-Gibbs entropy itself. For instance, \citet{LucenaSilvaTsallis1995} optimize the nonadditive entropic functional $S_{q}$ (see (\ref{EqSq}) below) under the constraint of Curado and Tsallis \cite{CuradoTsallis_1991} (i.e., $\sum^{W}_{i=1}p_{i}^{q}e_{i}=U$ instead of the linear $\sum^{W}_{i=1}p_{i}e_{i}=U$ in (\ref{maxProblem}) -- see \cite{DogniniTsallis2026} for a unified discussion of these variants) and show that, for $q<1$, the specific heat of the hydrogen atom in free space becomes a computable quantity.

This paper differs from these previous works in three aspects: the nonadditive entropic functional $S_{q}$ is maximized with $q\in(0,1)$ and a linear mean-energy constraint (i.e., the best-behaved case from \cite{DogniniTsallis2026} is adopted); the structural parameter $\beta_{q}$ depends only on the relative energy spectrum, whereas the temperature follows from the Clausius relation, possibly after a rescaling; and the hydrogen atom is treated as the limit of truncated systems (therefore, the truncation process is only instrumental, and the entire spectrum is later recovered), which yields a two-tier limiting optimal distribution for every $q\in(0,1)$. Its structure is that of the two grouped levels of \citet{DAmmandoColonnaPietanzaCapitelli2010}, with the upper level at the ionization threshold, but it arises here as an exact limit rather than as an approximation.

The other systems of Section \ref{secGen} have also been studied within the $S_{q}$ framework. \citet{ItoTsallis1989} compute the specific heat of the harmonic oscillator within the original formulation of \cite{Tsallis1988}, whose energy constraint is also linear, so that their distribution has the same ${(2-q)}$-exponential form as (\ref{eqGenDistribution}). However, they identify the parameter of the ${(2-q)}$-exponential with $1/kT$ instead of obtaining the temperature from the Clausius relation (\ref{eqTemp}), thus leading to different thermodynamic results. Also, \citet{SilvaTsallisCurado1993} study the one-dimensional free particle, both quantum and classical, under the constraint of \citet{CuradoTsallis_1991}. 

Furthermore, \citet{SilvaPlastinoLima1998} obtain, through a generalization of Maxwell's factorization hypothesis, a velocity distribution with the same functional form as (\ref{eqMBGen}), but with its parameter identified with $1/kT$ and without the corresponding relation to the mean energy (i.e., without the thermodynamics); here, the optimization yields $\beta_{q}=3/2qU$, and the rescaled temperature (\ref{eqMBTemp}) follows from the cell volume in velocity space.

\section{The generalized optimization problem}\label{sec2}

Throughout this paper, I consider the nonadditive entropic functional $S_{q}:\mathbb{R}^{W}_{+}\rightarrow\mathbb{R}$ \cite{Tsallis1988},
\begin{eqnarray}\label{EqSq}
    S_{q}(p)= k\sum^{W}_{i=1} p_{i}\ln_{q}\frac{1}{p_{i}},
\end{eqnarray}
with $q\in(0,1)$, where $\ln_{q}z=(z^{1-q}-1)/(1-q)$, $z>0$, is the $q$-logarithm, and $S_{1}=\lim_{q\rightarrow1}S_{q}=S_{BG}$. I recall the $q$-exponential $\exp_{q}z=[1+(1-q)z]_{+}^{1/(1-q)}$, $z\in\mathbb{R}$, with $[z]_{+}=\max\{z,0\}$. The generalized optimization problem studied in this paper is
\begin{eqnarray}\label{maxProblem}
\max& S_{q}(p)\\
\text{s.t.}& p_{i}\geq0, 1\leq i\leq W\nonumber\\
&\sum^{W}_{i=1}p_{i}=1\nonumber\\
&\sum^{W}_{i=1}p_{i}E_{i}=0,\nonumber
\end{eqnarray}
for $q\in(0,1]$, i.e., the mean-energy constraint is always \emph{linear}\footnote{In the notation of \citet{DogniniTsallis2026}, who studied the more general constraint $\sum_{i}p_{i}^{q'}E_{i}=0$, $q'>0$, this is the case $q'=1$.}, exactly as in (\ref{maxProblemBG}).

The next theorem gathers results of \cite{DogniniTsallis2026} (Theorems 1 and 2, Eqs.~(20) and (22), and Lemma 4 therein), and, to clarify the notation, $g_{q}:\mathbb{R}\times(e_{1},e_{W})\rightarrow[0,+\infty]$, $q\in(0,1)$, is given by $g_{q}(\beta,U)=\sum^{W}_{i=1}(\exp_{2-q}(-\beta E_{i}))^{q}$.

\begin{theorem}[Dognini and Tsallis \cite{DogniniTsallis2026}]\label{theoDT}
Let $q\in(0,1)$ and $U\in(e_{1},e_{W})$. Then, (\ref{maxProblem}) has a unique solution $\tilde{p}\in\mathbb{R}^{W}_{++}$, given by
\begin{eqnarray}\label{eqGenDistribution}
    \tilde{p}_{i}=\frac{\exp_{2-q}(-\beta_{q} E_{i})}{\sum^{W}_{j=1}\exp_{2-q}(-\beta_{q} E_{j})}
\end{eqnarray}
for $1\leq i\leq W$, where the structural parameter $\beta_{q}\in\mathbb{R}$ is the unique solution of
\begin{eqnarray}\label{eqGenBeta}
    \sum^{W}_{i=1}\exp_{2-q}(-\beta_{q} E_{i})E_{i}=0.
\end{eqnarray}
The entropy is given by $S_{q}(\tilde{p})=k\ln_{q}g_{q}(\beta_{q},U)$, and the generalized temperature is
\begin{eqnarray}
    \frac{1}{T_{q}}\equiv\frac{\partial S_{q}}{\partial U}=kqg_{q}(\beta_{q},U)^{1-q}\beta_{q}.\label{eqTemp}
\end{eqnarray}
Furthermore, $\beta_{q}\rightarrow\beta_{BG}$ as $q\rightarrow1$, with $\beta_{BG}$ the solution of (\ref{eqBGBeta}), so that (\ref{eqGenDistribution}) recovers (\ref{eqBGDistribution}).
\end{theorem}

The next corollary, which restates an identity from \cite{DogniniTsallis2026}, follows straightforwardly from the definitions and will be extensively used in Section \ref{secGen}.

\begin{corollary}\label{corSum}
    Let $q\in(0,1)$, $U\in(e_{1},e_{W})$, and let $\tilde{p}\in\mathbb{R}^{W}_{++}$ be the unique solution of (\ref{maxProblem}). Then,
    \begin{eqnarray}\label{eqSumPq}
    \sum^{W}_{i=1}\exp_{2-q}(-\beta_{q}E_{i})=g_{q}(\beta_{q},U),
\end{eqnarray}
so that $\tilde{p}_{i}=\exp_{2-q}(-\beta_{q}E_{i})/g_{q}(\beta_{q},U)$, $1\leq i\leq W$, and $\sum_{i}\tilde{p}^{q}_{i}=g_{q}(\beta_{q},U)^{1-q}$.
\end{corollary}

The following proposition (Proposition 5 of \cite{DogniniTsallis2026}) clarifies the sign of the structural parameter $\beta_{q}\in\mathbb{R}$ and shows that this framework preserves the Third Law of Thermodynamics.

\begin{proposition}[Dognini and Tsallis \cite{DogniniTsallis2026}]\label{propDT}
Let $q\in(0,1)$ and $0=e_{1}<e_{W}$. The value function $U\mapsto S_{q}(\tilde{p}(U))$, $U\in(0,e_{W})$, is smooth and strictly concave, it attains its maximum $k\ln_{q}W$ at $U=\sum^{W}_{i=1}e_{i}/W$, $\lim_{U\rightarrow0}S_{q}(\tilde{p}(U))=k\ln_{q}\Omega(e_{1})$, and $\lim_{U\rightarrow0}\partial S_{q}(\tilde{p}(U))/\partial U=+\infty$. In particular, $\beta_{q}>0$ (i.e., $T_{q}>0$) if, and only if, $U<\sum^{W}_{i=1}e_{i}/W$.
\end{proposition}

All the systems studied in this paper are modeled through a limiting process over the energy spectrum. As remarked in \cite{DogniniTsallis2026}, the right-hand side of (\ref{eqTemp}) may then diverge or tend to zero when $W\rightarrow\infty$, and an adequate rescaling becomes necessary to ensure well-behaved thermodynamics. Let $s(W)>0$ be a factor that depends on the number of possible states of the system (the ``s'' stands for ``scale''). The rescaled functional $S_{q}/s(W)$ has the same optimal distribution (\ref{eqGenDistribution}) and defines the rescaled temperature as the following limit
\begin{eqnarray}\label{eqTempRescaled}
    T_{q, s}=\lim_{W\rightarrow\infty} s(W)T_{q}=\lim_{W\rightarrow\infty}\frac{s(W)}{kq g_{q}(\beta_{q},U)^{1-q}\beta_{q}}.
\end{eqnarray}
A natural choice, suggested in \cite{DogniniTsallis2026}, is $s(W)=\ln_{q}W=S_{q}(1/W,\ldots,1/W)/k$ (i.e., the dimensionless maximum entropy). We will see that no rescaling is needed for the harmonic oscillator with $q>1/2$ and the one-dimensional box with $q>1/3$ (Subsections \ref{secGho} and \ref{secGbox}), that the level spacing and the cell volume in velocity space furnish the scale factors for spectra that become continuous (Subsections \ref{secGunif} and \ref{secGmb}), and that $s(W)=\ln_{q}W$ is the adequate one for the hydrogen atom (Subsection \ref{secGh}).

Before stating the last result of this section, I need extra notation. Let the energy spectrum \emph{without degeneracy} be $0=e_{1}<e_{2}<\ldots$, with degeneracy function $\Omega(e_{n})$, $n\geq1$, and $L(U)=\max\{n\geq1\mid e_{n}\leq U\}$, $0<U<\sup\{e_{n}\}_{n\geq1}$. For $N>L(U)$, consider (\ref{maxProblem}) with only the first $N$ energy values $e_{1}<\ldots<e_{N}$ (i.e., with $W(N)=\sum^{N}_{n=1}\Omega(e_{n})$ states), and let $\beta^{N}_{q}$ be the corresponding structural parameter, $\tilde{p}^{N}_{n}$ the optimal probability of a \emph{single state} with energy $e_{n}$, and $\tilde{P}^{N}_{n}=\Omega(e_{n})\tilde{p}^{N}_{n}$ the optimal probability of the energy level $e_{n}$, $1\leq n\leq N$. By (\ref{eqGenBeta}), $\beta^{N}_{q}$ satisfies
\begin{eqnarray}\label{eqBetaTrunc}
    \sum^{N}_{n=1}\Omega(e_{n})\exp_{2-q}(-\beta^{N}_{q}E_{n})E_{n}=0.
\end{eqnarray}
Moreover, by (\ref{eqGenDistribution}) and (\ref{eqSumPq}), the generalized temperature (\ref{eqTemp}) of the truncated system can be written as
\begin{eqnarray}\label{eqTempTrunc}
    kT^{N}_{q}=\frac{1-(1-q)\beta^{N}_{q}U}{q\beta^{N}_{q}}\biggl(\frac{\tilde{P}^{N}_{1}}{\Omega(e_{1})}\biggr)^{1-q},
\end{eqnarray}
since $\exp_{2-q}(\beta^{N}_{q}U)=(1-(1-q)\beta^{N}_{q}U)^{-1/(1-q)}$ and
\begin{eqnarray*}
g_{q}(\beta^{N}_{q},U)=\sum^{N}_{n=1}\Omega(e_{n})\exp_{2-q}(-\beta^{N}_{q}E_{n})=\frac{\Omega(e_{1})\exp_{2-q}(\beta^{N}_{q}U)}{\tilde{P}^{N}_{1}}.
\end{eqnarray*}

\begin{theorem}\label{theoTrunc}
Let $q\in(0,1)$, $0<U<\sup\{e_{n}\}_{n\geq1}$ and $\beta^{N}_{q}$, $N>L(U)$, be as defined above. Then:
\begin{enumerate}
    \item[(i)] $\beta^{N}_{q}<1/(1-q)U$ and, writing $\delta^{N}=1/(1-q)U-\beta^{N}_{q}>0$,
    \begin{eqnarray}\label{eqRatios}
        \frac{\tilde{p}^{N}_{n}}{\tilde{p}^{N}_{1}}=\biggl(\frac{U\delta^{N}}{U\delta^{N}+\beta^{N}_{q}e_{n}}\biggr)^{\frac{1}{1-q}},
    \end{eqnarray}
    for $1\leq n\leq N$;
    \item[(ii)] there is $N_{1}>L(U)$ such that $\sum^{N}_{n=1}\Omega(e_{n})E_{n}>0$ for $N\geq N_{1}$, the sequence $\{\beta^{N}_{q}\}_{N\geq N_{1}}$ is strictly increasing and positive, and there is $\beta_{q}\in(0,1/(1-q)U]$ such that $\lim_{N\rightarrow\infty}\beta^{N}_{q}=\beta_{q}$;
    \item[(iii)] if $\beta_{q}=1/(1-q)U$ (i.e., $\lim_{N\rightarrow\infty}\delta^{N}=0$), then $\lim_{N\rightarrow\infty}\tilde{P}^{N}_{n}=0$ for $n\geq2$, $\lim_{N\rightarrow\infty}\tilde{P}^{N}_{1}=1-U/e_{\infty}$, with $e_{\infty}=\sup\{e_{n}\}_{n\geq1}\in(U,+\infty]$ (in particular, $\lim_{N\rightarrow\infty}\tilde{P}^{N}_{1}=1$ if the spectrum is unbounded), and the generalized temperature $T^{N}_{q}$ given by (\ref{eqTemp}) for the truncated system satisfies $\lim_{N\rightarrow\infty}T^{N}_{q}=0$;
    \item[(iv)] $\beta_{q}=1/(1-q)U$ if, and only if, $\Psi(U)=+\infty$, where
    \begin{eqnarray}\label{eqPsi}
        \Psi(U)=\sum^{\infty}_{n=L(U)+1}\Omega(e_{n})\biggl(\frac{e_{n}}{U}\biggr)^{\frac{1}{q-1}}(e_{n}-U);
    \end{eqnarray}
    \item[(v)] if $\beta_{q}<1/(1-q)U$ (i.e., $\delta=\lim_{N\rightarrow\infty}\delta^{N}>0$), then, for every $n\geq1$, $\tilde{p}^{N}_{n}$ converges to $\tilde{p}_{n}=\tilde{p}_{1}r_{n}$ as $N\rightarrow\infty$, with
    \begin{eqnarray*}
        r_{n}=\biggl(\frac{U\delta}{U\delta+\beta_{q} e_{n}}\biggr)^{\frac{1}{1-q}}\qquad\text{and}\qquad\tilde{p}_{1}=\biggl(\sum^{\infty}_{n=1}\Omega(e_{n})r_{n}\biggr)^{-1},
    \end{eqnarray*}
    the limiting distribution satisfies $\sum^{\infty}_{n=1}\Omega(e_{n})\tilde{p}_{n}=1$ and $\sum^{\infty}_{n=1}\Omega(e_{n})\tilde{p}_{n}e_{n}=U$, $\beta_{q}$ is the unique solution in $(0,1/(1-q)U)$ of $\sum^{\infty}_{n=1}\Omega(e_{n})\exp_{2-q}(-\beta_{q} E_{n})E_{n}=0$, and
    \begin{eqnarray}\label{eqTempLim}
        \lim_{N\rightarrow\infty}kT^{N}_{q}=\frac{(1-q)U\delta\,\tilde{p}_{1}^{1-q}}{q\beta_{q}}\in(0,+\infty).
    \end{eqnarray}
\end{enumerate}
\end{theorem}

\section{Generalized Boltzmann-Gibbs distributions}\label{secGen}

In this section, I solve (\ref{maxProblem}) for several textbook systems whose description requires a limiting process over the energy spectrum and describe the corresponding thermodynamics\footnote{Although not used in this paper, \citet{DogniniTsallis2026} also define a $q$-generalized Helmholtz-like energy and partition function, with the latter based on a generalized product, known as \textit{$q$-product} (see \cite{DogniniTsallis2025} for a recent discussion on these algebraic structures).}, with the Boltzmann-Gibbs results recovered as $q\rightarrow1$.

\subsection{The hydrogen atom in free space}\label{secGh}

The energy spectrum of the hydrogen atom in free space without degeneracy is given by (\ref{eqHSpectrum}), with $\Omega(e_{n})=2n^{2}$, $n\geq1$, so that the truncated system has
\begin{eqnarray*}
    W(N)=\sum^{N}_{n=1}2n^{2}=\frac{N(N+1)(2N+1)}{3}
\end{eqnarray*}
states. For $U\in(0,e_{ion})$, I consider (\ref{maxProblem}) with only the first $N>L(U)$ energy values. For comparison, I first state the result behind the Electronic Partition Function Paradox, which shows that the Boltzmann-Gibbs framework leads to degenerate thermodynamics for the hydrogen atom in free space.

\begin{proposition}\label{propParadox}
Let $q=1$, $U\in(0,e_{ion})$, and $T^{N}_{BG}=1/k\beta^{N}_{BG}$, $N>L(U)$, with $\beta^{N}_{BG}$ the solution of (\ref{eqBGBeta}) for the truncated spectrum. Then, $\lim_{N\rightarrow\infty}\beta^{N}_{BG}=+\infty$ and $\lim_{N\rightarrow\infty}T^{N}_{BG}=0$. Moreover,
\begin{eqnarray}\label{eqBGRescaled}
    \lim_{N\rightarrow\infty}\bigl(\beta^{N}_{BG}e_{ion}-\ln W(N)\bigr)=\ln\frac{e_{ion}-U}{2U}\qquad\text{and}\qquad\lim_{N\rightarrow\infty}\ln W(N)\,T^{N}_{BG}=\frac{e_{ion}}{k},
\end{eqnarray}
and the corresponding optimal probabilities of the energy levels satisfy $\lim_{N\rightarrow\infty}\tilde{P}^{N}_{1}=1-U/e_{ion}$ and $\lim_{N\rightarrow\infty}\tilde{P}^{N}_{n}=0$, $n\geq2$.
\end{proposition}

Therefore, the Boltzmann-Gibbs temperature vanishes and, even if rescaled by $\ln W(N)$ (i.e., by the scale factor $\ln_{q}W(N)$ with $q=1$), it converges to the same value for every mean energy $U\in(0,e_{ion})$, so that the temperature does not determine the mean energy.

Next, let $q\in(0,1)$. The series (\ref{eqPsi}) diverges for every $q\in(0,1)$, since its general term, $2n^{2}(e_{n}/U)^{\frac{1}{q-1}}(e_{n}-U)$, grows as $n^{2}$. Therefore, by Theorem \ref{theoTrunc}(iv), $\delta^{N}\rightarrow0$ and, by Theorem \ref{theoTrunc}(iii), the generalized temperature (\ref{eqTemp}) of the truncated system also tends to zero as $N\rightarrow\infty$. In this case, the scale factor $s(W(N))=\ln_{q}W(N)$ in (\ref{eqTempRescaled}) furnishes finite temperatures, as the next result will show. Before stating it, however, let the rescaled temperature of the truncated system be given by
\begin{eqnarray}\label{eqHTempN}
    T^{N}_{q,s}=\frac{\ln_{q}W(N)}{kq\,g_{q}(\beta^{N}_{q},U)^{1-q}\beta^{N}_{q}},
\end{eqnarray}
for $q\in(0,1)$, $U\in(0,e_{ion})$ and $N>L(U)$, and let $S^{N}_{q}=k\ln_{q}g_{q}(\beta^{N}_{q},U)$ be the corresponding entropy (see Theorem \ref{theoDT}). I proceed with the main theorem of the paper.

\begin{theorem}\label{theoH}
Let $q\in(0,1)$, $U\in(0,e_{ion})$, and $\tilde{P}^{N}_{n}$, $T^{N}_{q,s}$ and $S^{N}_{q}$, $N>L(U)$, be as defined above. Then,
\begin{eqnarray*}
    \lim_{N\rightarrow\infty}\beta^{N}_{q}=\frac{1}{(1-q)U},\qquad\lim_{N\rightarrow\infty}N^{3(1-q)}\delta^{N}=\frac{3^{1-q}e_{ion}}{(1-q)(e_{ion}-U)^{1-q}U^{1+q}},
\end{eqnarray*}
$\lim_{N\rightarrow\infty}\tilde{P}^{N}_{1}=1-U/e_{ion}$, $\lim_{N\rightarrow\infty}\tilde{P}^{N}_{n}=0$, $n\geq2$, and
\begin{eqnarray*}
    T_{q,s}=\lim_{N\rightarrow\infty}T^{N}_{q,s}=\frac{U^{1-q}e^{q}_{ion}}{kq},\qquad S_{q,s}=\lim_{N\rightarrow\infty}\frac{S^{N}_{q}}{\ln_{q}W(N)}=k\biggl(\frac{U}{e_{ion}}\biggr)^{q}.
\end{eqnarray*}
\end{theorem}

Although the probability of each energy level $n\geq2$ given by Theorem \ref{theoH} tends to zero, this only means that the remaining probability out of level $n=1$ (i.e., $U/e_{ion}$) moves, as $N\rightarrow\infty$, towards the upper levels and accumulates at $e_{ion}$. This is made precise by the following corollary.

\begin{corollary}\label{corH}
Let $\{H^{N}\}_{N>L(U)}$ be the sequence of random variables representing the energy of the hydrogen atom with mean energy $U\in(0,e_{ion})$ and a truncated energy spectrum (i.e., $\mathbb{P}(H^{N}=e_{n})=\tilde{P}^{N}_{n}$, $1\leq n\leq N$). Then, $H^{N}\overset{d}{\rightarrow}H$, with $\mathbb{P}(H=e_{ion})=U/e_{ion}$ and $\mathbb{P}(H=0)=1-U/e_{ion}$.
\end{corollary}

Therefore, the limiting optimal distribution is a two-tier one, concentrated on the ground level and on the ionization threshold, with mean energy $\mathbb{E}(H)=U$. Remarkably, this behavior does not depend on the actual value of $q\in(0,1]$ (see Proposition \ref{propParadox} for $q=1$), which only enters the temperature. Defining the critical ionization temperature as $T_{ion}=e_{ion}/kq$, we have (see Figure \ref{figH}(a))
\begin{eqnarray}\label{eqHTemp2}
    T_{q,s}=T_{ion}\biggl(\frac{U}{e_{ion}}\biggr)^{1-q}\qquad\text{and}\qquad U=e_{ion}\biggl(\frac{T_{q,s}}{T_{ion}}\biggr)^{\frac{1}{1-q}},
\end{eqnarray}
so that the specific heat is
\begin{eqnarray}\label{eqHC}
    C=\frac{dU}{dT_{q,s}}=\frac{kq}{1-q}\biggl(\frac{T_{q,s}}{T_{ion}}\biggr)^{\frac{q}{1-q}},
\end{eqnarray}
which vanishes as $T_{q,s}\rightarrow0$, in accordance with the Third Law of Thermodynamics (see Figure \ref{figH}(b)). Moreover, by Theorem \ref{theoH}, the rescaled entropy is $S_{q,s}=k(U/e_{ion})^{q}$, so that $1/T_{q,s}=dS_{q,s}/dU$ (i.e., the Clausius relation is preserved in the limit) and $S_{q,s}\rightarrow0$ as $T_{q,s}\rightarrow0$.

Notice that, as $q\rightarrow1$, (\ref{eqHTemp2}) yields $T_{q,s}\rightarrow e_{ion}/k$ for every $U\in(0,e_{ion})$, which is precisely the rescaled Boltzmann-Gibbs temperature of Proposition \ref{propParadox}. In other words, the limiting distribution of the Boltzmann-Gibbs case is the same two-tier one, but its rescaled temperature is the same for every mean energy and, therefore, does not determine it. It is only for $q<1$ that (\ref{eqHTemp2}) is a one-to-one relation between $U$ and $T_{q,s}$, thus providing well-defined thermodynamics.

\begin{figure}[H]
    \centering
    \includegraphics[width=\textwidth]{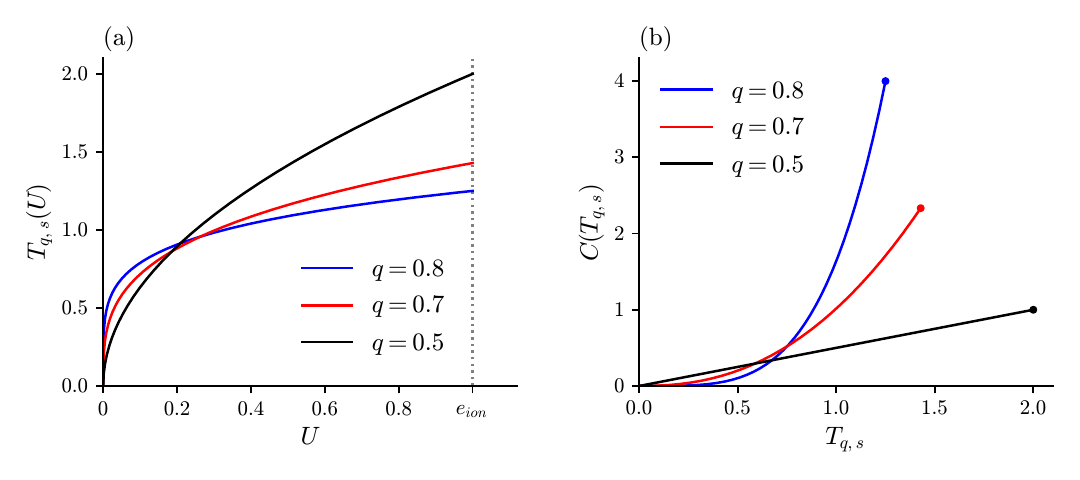}
    \caption{Hydrogen atom in free space: (a) temperature $T_{q,s}$ as a function of the mean energy $U\in(0,e_{ion})$, given by (\ref{eqHTemp2}); (b) specific heat $C$ as a function of the temperature $T_{q,s}\in(0,T_{ion})$, given by (\ref{eqHC}), with the dots marking its value $kq/(1-q)$ at $T_{q,s}=T_{ion}$. In both panels, $q=0.5,0.7,0.8$, $U$ is in units of $e_{ion}$, $T_{q,s}$ in units of $e_{ion}/k$ and $C$ in units of $k$.}
    \label{figH}
\end{figure}

\subsection{An energy spectrum that uniformly approaches a continuous, unbounded limit}\label{secGunif}

Consider the following energy spectrum without degeneracy
\begin{eqnarray*}
    e_{n}(e_{\max},N)=\frac{(n-1)e_{\max}}{N-1},
\end{eqnarray*}
for $1\leq n\leq N$, $N\geq2$, $e_{\max}>0$ (i.e., the $N\geq2$ distinct energy values are equispaced between $0$ and $e_{\max}>0$). Also, let the common degeneracy be $\Omega(e_{n})=\Omega_{1}\geq1$, $1\leq n\leq N$, so that $W=\Omega_{1}N$ (i.e., the total number of states equals the common degeneracy times the number of different energy levels).

I am interested in the limit $N\rightarrow\infty$ (the spectrum becomes continuous) followed by $e_{\max}\rightarrow\infty$ (the spectrum becomes unbounded). Let $q\in(1/2,1)$, $0<U<e_{\max}$, and $\beta_{q}(U,e_{\max},N)$ satisfy (\ref{eqGenBeta}), i.e.,
\begin{eqnarray}\label{eqUnifDiscrete}
    \frac{e_{\max}}{N-1}\sum^{N}_{n=1}\exp_{2-q}(-\beta_{q}(U,e_{\max},N)(e_{n}-U))(e_{n}-U)=0.
\end{eqnarray}
Assuming $\lim_{N\rightarrow\infty}\beta_{q}(U,e_{\max},N)=\beta_{q}(U,e_{\max})$, $\lim_{e_{\max}\rightarrow\infty}\beta_{q}(U,e_{\max})=\beta_{q}(U)$ and the appropriate convergence, (\ref{eqUnifDiscrete}) can be seen as a Riemann sum that yields
\begin{eqnarray}\label{eqUnifGen1}
    \int^{\infty}_{0}\exp_{2-q}(-\beta_{q}(U)(e-U))(e-U)\,de=0.
\end{eqnarray}
The integrand decays as a power law with exponent $-q/(1-q)$ for large energies, so that (\ref{eqUnifGen1}) is well defined only for $q>1/2$, which motivates the restriction above. Since $\beta_{q}(U)>0$ by Proposition \ref{propDT},
\begin{eqnarray*}
    \int^{\infty}_{0}\exp_{2-q}(-\beta_{q}(U)(e-U))\,de=\frac{(1-(1-q)\beta_{q}(U)U)^{\frac{q}{q-1}}}{q\beta_{q}(U)},\\
    \int^{\infty}_{0}\exp_{2-q}(-\beta_{q}(U)(e-U))^{q}\,de=\frac{(1-(1-q)\beta_{q}(U)U)^{\frac{2q-1}{q-1}}}{(2q-1)\beta_{q}(U)},
\end{eqnarray*}
and, since (\ref{eqUnifGen1}) is equivalent to the equality of these two integrals (see (\ref{eqSumPq})),
\begin{eqnarray}\label{eqUnifBeta}
    \beta_{q}(U)=\frac{1}{qU},
\end{eqnarray}
which recovers the Boltzmann-Gibbs result $\beta_{BG}(U)=1/U$ as $q\rightarrow1$. 

Let $\Delta e=e_{\max}/(N-1)$ be the level spacing and let the \textit{density estimate} $\rho(\cdot,U,e_{\max},N):[0,e_{\max}+\Delta e)\rightarrow\mathbb{R}_{+}$ be given by $\rho(e,U,e_{\max},N)=\tilde{P}^{N}_{n(e)}/\Delta e$ for $e_{n(e)}\leq e<e_{n(e)}+\Delta e$, $1\leq n(e)\leq N$, with $\tilde{P}^{N}_{n}=\Omega_{1}\tilde{p}^{N}_{n}$ the optimal probability of the energy level $e_{n}$, so that $\int_{e_{n}}^{e_{n}+\Delta e}\rho\,dx=\tilde{P}^{N}_{n}$, $1\leq n\leq N$, and $\rho$ integrates to one. Then, the limiting distribution over the energy spectrum is
\begin{eqnarray*}
    \rho(e,U)&=&\lim_{e_{\max}\rightarrow\infty}\Bigl(\lim_{N\rightarrow\infty}\rho(e,U,e_{\max},N)\Bigr)=\frac{\exp_{2-q}(-\beta_{q}(U)(e-U))}{\int^{\infty}_{0}\exp_{2-q}(-\beta_{q}(U)(x-U))\,dx}\\
    &=&\frac{1}{U}\biggl(\frac{2q-1}{q}\biggr)^{\frac{q}{1-q}}\exp_{2-q}\biggl(\frac{1}{q}-\frac{e}{qU}\biggr),
\end{eqnarray*}
which recovers the Boltzmann-Gibbs density $\exp(-e/U)/U$ as $q\rightarrow1$.

I now turn to the temperature. To ease notation, let $\Delta^{N}=e_{\max}/\Omega_{1}(N-1)$ and, using (\ref{eqSumPq}),
\begin{eqnarray*}
I^{N}=\Delta^{N}\,g_{q}(\beta_{q}(U,e_{\max},N),U)=\frac{e_{\max}}{N-1}\sum^{N}_{n=1}\exp_{2-q}(-\beta_{q}(U,e_{\max},N)E_{n}),
\end{eqnarray*}
so that, assuming the appropriate convergence, we have
\begin{eqnarray}\label{eqIConvergence}
    I=\lim_{e_{\max}\rightarrow\infty}\Bigl(\lim_{N\rightarrow\infty}I^{N}\Bigr)=\int^{\infty}_{0}\exp_{2-q}(-\beta_{q}(U)(x-U))\,dx=U\biggl(\frac{q}{2q-1}\biggr)^{\frac{q}{1-q}}.
\end{eqnarray}
Let $\varepsilon>0$ be a reference energy (e.g., the energy unit or resolution adopted) and consider the following scale factor
\begin{eqnarray}\label{eqUnifLambda}
    s(W(N))=\biggl(\frac{\varepsilon}{\Delta^{N}}\biggr)^{1-q}.
\end{eqnarray}
In particular, notice that $s(W(N))\rightarrow1$ as $q\rightarrow1$.\footnote{Since $\Delta^{N}=e_{\max}/(W(N)-\Omega_{1})$, the scale factor (\ref{eqUnifLambda}) also depends on $e_{\max}$ and $\Omega_{1}$; I keep the notation $s(W(N))$ for simplicity. The same remark applies to Subsection \ref{secGmb}, where the scale factor also depends on $v_{\max}$.} Then, (\ref{eqTempRescaled}), (\ref{eqUnifBeta}) and (\ref{eqIConvergence}) imply
\begin{eqnarray*}
    kT_{q,s}=\lim_{e_{\max}\rightarrow\infty}\Bigl(\lim_{N\rightarrow\infty}\frac{\varepsilon^{1-q}}{q\beta_{q}(U,e_{\max},N)(I^{N})^{1-q}}\Bigr)=\frac{\varepsilon^{1-q}}{q\beta_{q}(U)I^{1-q}}=\biggl(\frac{2q-1}{q}\biggr)^{q}\varepsilon^{1-q}U^{q},
\end{eqnarray*}
so that
\begin{eqnarray*}
    U=\frac{q}{2q-1}\,\varepsilon\biggl(\frac{kT_{q,s}}{\varepsilon}\biggr)^{\frac{1}{q}},
\end{eqnarray*}
which recovers the Boltzmann-Gibbs relation $U=kT$ as $q\rightarrow1$. Finally, the specific heat becomes
\begin{eqnarray*}
    C=\frac{dU}{dT_{q,s}}=\frac{k}{2q-1}\biggl(\frac{kT_{q,s}}{\varepsilon}\biggr)^{\frac{1-q}{q}}.
\end{eqnarray*}

\subsection{The Maxwell-Boltzmann distribution of velocities}\label{secGmb}

I now consider the velocity distribution of a particle of mass $m>0$ in a dilute gas, whose energy is $e(\mathbf{v})=m\vert\mathbf{v}\vert^{2}/2$, $\mathbf{v}\in\mathbb{R}^{3}$. Let the cube $[-v_{\max},v_{\max}]^{3}$ be partitioned into $W(N)=N^{3}$ cells of side $\Delta v=2v_{\max}/N$ and volume $(\Delta v)^{3}$, and let the $i$-th state correspond to the $i$-th cell, with velocity $\mathbf{v}_{i}$ at its center and energy $e_{i}=m\vert\mathbf{v}_{i}\vert^{2}/2$, $1\leq i\leq W(N)$.

Next, I solve (\ref{maxProblem}) for this finite spectrum and, as in Subsection \ref{secGunif}, assume the appropriate convergence when $N\rightarrow\infty$ (the spectrum becomes continuous) and then $v_{\max}\rightarrow\infty$ (the spectrum becomes unbounded), so that sums over states become integrals over $\mathbb{R}^{3}$ through $\sum_{i}(\Delta v)^{3}(\cdot)\rightarrow\int d^{3}\mathbf{v}\,(\cdot)$, and $\rho(\mathbf{v})=\lim\tilde{p}_{i(\mathbf{v})}/(\Delta v)^{3}$ is the limiting probability density.

Let $q\in(0,1)$ and $\nu=1/(1-q)$. By (\ref{eqGenDistribution}), $\rho(\mathbf{v})\propto\exp_{2-q}(-\beta_{q}(e(\mathbf{v})-U))=[1-(1-q)\beta_{q}U+(1-q)\beta_{q}e(\mathbf{v})]^{-\nu}$, which decays as $\vert\mathbf{v}\vert^{-2\nu}$. The normalization requires $2\nu>3$ (i.e., $q>1/3$), and the energy constraint (\ref{eqGenBeta}) requires $2\nu>5$ (i.e., $q>3/5$). Therefore, I assume $q>3/5$ and $a=1-(1-q)\beta_{q}U>0$. 

By (\ref{eqSumPq}), (\ref{eqGenBeta}) is equivalent to $\int\exp_{2-q}(-\beta_{q}(e-U))\,d^{3}\mathbf{v}=\int\exp_{2-q}(-\beta_{q}(e-U))^{q}\,d^{3}\mathbf{v}$, i.e., to
\begin{eqnarray*}
    \int^{\infty}_{0}\sqrt{e}\,(a+(1-q)\beta_{q}e)^{-\nu}de=\int^{\infty}_{0}\sqrt{e}\,(a+(1-q)\beta_{q}e)^{-(\nu-1)}de.
\end{eqnarray*}
Since $\int^{\infty}_{0}\sqrt{e}\,(a+ce)^{-\mu}de=a^{\frac{3}{2}-\mu}c^{-\frac{3}{2}}\Gamma(\frac{3}{2})\Gamma(\mu-\frac{3}{2})/\Gamma(\mu)$, $\mu>3/2$, the equality above becomes $a^{-1}=\Gamma(\nu-\frac{5}{2})\Gamma(\nu)/\Gamma(\nu-1)\Gamma(\nu-\frac{3}{2})=(\nu-1)/(\nu-\frac{5}{2})$, i.e., $a=1-3(1-q)/2q=(5q-3)/2q$, so that
\begin{eqnarray}\label{eqMBBeta}
    \beta_{q}=\frac{3}{2qU},
\end{eqnarray}
which recovers the Boltzmann-Gibbs result $\beta_{BG}=3/2U$ as $q\rightarrow1$. Thus, $\exp_{2-q}(-\beta_{q}{(e-U)})=a^{-\nu}[{1+(1-q)\beta_{q}e/a}]^{-\nu}$ and the generalized Maxwell-Boltzmann distribution is the ${(2-q)}$-Gaussian
\begin{eqnarray}\label{eqMBGen}
    \rho(\mathbf{v})&=&\biggl(\frac{b}{\pi}\biggr)^{\frac{3}{2}}\frac{\Gamma(\nu)}{\Gamma(\nu-\frac{3}{2})}\bigl(1+b\vert\mathbf{v}\vert^{2}\bigr)^{-\frac{1}{1-q}}\nonumber\\
    &=&\biggl(\frac{b}{\pi}\biggr)^{\frac{3}{2}}\frac{\Gamma(\nu)}{\Gamma(\nu-\frac{3}{2})}\exp_{2-q}\biggl(-\frac{3}{5q-3}\,\frac{m\vert\mathbf{v}\vert^{2}}{2U}\biggr),
\end{eqnarray}
with $b=3(1-q)m/2(5q-3)U$, which recovers the Maxwell-Boltzmann distribution
\begin{eqnarray}\label{eqMB}
    \rho(\mathbf{v})=\biggl(\frac{3m}{4\pi U}\biggr)^{\frac{3}{2}}\exp\biggl(-\frac{3m\vert\mathbf{v}\vert^{2}}{4U}\biggr)
\end{eqnarray}
as $q\rightarrow1$ (since $(1-q)^{3/2}\Gamma(\nu)/\Gamma(\nu-3/2)\rightarrow1$). Figure \ref{figMB} depicts the distribution of speeds $F(v)=4\pi v^{2}\rho(\mathbf{v})$ for a fixed mean energy, revealing the power-law tails $F(v)\sim v^{-2q/(1-q)}$.

I now turn to the temperature. Let $\varepsilon>0$ be a reference energy and $V_{\varepsilon}=(2\pi\varepsilon/m)^{3/2}$ the corresponding reference volume in velocity space, and let $s(W(N))=(V_{\varepsilon}/(\Delta v)^{3})^{1-q}$, the analogue of (\ref{eqUnifLambda}). Then, with $I=\int\exp_{2-q}(-\beta_{q}(e(\mathbf{v})-U))\,d^{3}\mathbf{v}$, (\ref{eqTempRescaled}), (\ref{eqMBBeta}) and (\ref{eqMBGen}) yield\footnote{In five detailed steps: (i) by (\ref{eqTempRescaled}) and $s(W(N))=(V_{\varepsilon}/(\Delta v)^{3})^{1-q}$, the rescaled temperature of the truncated system is
\begin{eqnarray*}
kT^{N}_{q,s}=\frac{(V_{\varepsilon}/(\Delta v)^{3})^{1-q}}{q\beta^{N}_{q}g_{q}(\beta^{N}_{q},U)^{1-q}}=\frac{V_{\varepsilon}^{1-q}}{q\beta^{N}_{q}\bigl((\Delta v)^{3}g_{q}(\beta^{N}_{q},U)\bigr)^{1-q}};
\end{eqnarray*}
(ii) by (\ref{eqSumPq}), $(\Delta v)^{3}g_{q}(\beta^{N}_{q},U)=\sum_{i}(\Delta v)^{3}\exp_{2-q}(-\beta^{N}_{q}E_{i})$ is a Riemann sum, so that, under the convergence assumptions above, $\beta^{N}_{q}\rightarrow\beta_{q}$ and
\begin{eqnarray*}
(\Delta v)^{3}g_{q}(\beta^{N}_{q},U)\rightarrow I=\int(a+(1-q)\beta_{q}e(\mathbf{v}))^{-\nu}\,d^{3}\mathbf{v};
\end{eqnarray*}
(iii) in spherical coordinates, with $e=mv^{2}/2$, we have $d^{3}\mathbf{v}=4\pi v^{2}dv=4\sqrt{2}\pi m^{-3/2}\sqrt{e}\,de$, so that the integral formula used above for $\beta_{q}$ (with $\mu=\nu$, $c=(1-q)\beta_{q}$ and $\Gamma(\frac{3}{2})=\sqrt{\pi}/2$) yields
\begin{eqnarray*}
I=\frac{4\sqrt{2}\pi}{m^{3/2}}\int^{\infty}_{0}\sqrt{e}\,(a+(1-q)\beta_{q}e)^{-\nu}de=\frac{2\sqrt{2}\pi^{\frac{3}{2}}a^{\frac{3}{2}-\nu}}{(m(1-q)\beta_{q})^{\frac{3}{2}}}\frac{\Gamma(\nu-\frac{3}{2})}{\Gamma(\nu)}=a^{-\nu}\biggl(\frac{\pi}{b}\biggr)^{\frac{3}{2}}\frac{\Gamma(\nu-\frac{3}{2})}{\Gamma(\nu)},
\end{eqnarray*}
with $b=(1-q)\beta_{q}m/2a=3(1-q)m/2(5q-3)U$, since $\beta_{q}=3/2qU$ and $a=(5q-3)/2q$ (this is consistent with the normalization constant in (\ref{eqMBGen}), since $\rho=\exp_{2-q}(-\beta_{q}(e-U))/I$); (iv) since $\nu(1-q)=1$ and $V_{\varepsilon}=(2\pi\varepsilon/m)^{3/2}$,
\begin{eqnarray*}
\frac{V_{\varepsilon}^{1-q}}{I^{1-q}}=a\biggl(\frac{2\varepsilon b}{m}\biggr)^{\frac{3(1-q)}{2}}\biggl(\frac{\Gamma(\nu)}{\Gamma(\nu-\frac{3}{2})}\biggr)^{1-q}=a\biggl(\frac{3}{5q-3}\biggr)^{\frac{3(1-q)}{2}}\biggl(\frac{\varepsilon}{U}\biggr)^{\frac{3(1-q)}{2}}\biggl(\frac{(1-q)^{\frac{3}{2}}\Gamma(\nu)}{\Gamma(\nu-\frac{3}{2})}\biggr)^{1-q};
\end{eqnarray*}
(v) finally, dividing by $q\beta_{q}=3/2U$ and using $a=(5q-3)/2q$, we obtain (\ref{eqMBTemp}).}\predisplaypenalty=0
\begin{eqnarray}\label{eqMBTemp}
    kT_{q,s}=\frac{V_{\varepsilon}^{1-q}}{q\beta_{q}I^{1-q}}=d_{q}\,\varepsilon^{\frac{3(1-q)}{2}}U^{\frac{3q-1}{2}}, 
\end{eqnarray}
with
\begin{eqnarray*}
    d_{q}=\frac{5q-3}{3q}\biggl(\frac{3}{5q-3}\biggr)^{\frac{3(1-q)}{2}}\biggl(\frac{(1-q)^{\frac{3}{2}}\Gamma(\nu)}{\Gamma(\nu-\frac{3}{2})}\biggr)^{1-q}.
\end{eqnarray*}
Notice that $d_{q}\rightarrow2/3$ as $q\rightarrow1$, thus recovering $U=3kT/2$. The corresponding specific heat, $C=dU/dT_{q,s}=2U/(3q-1)T_{q,s}$, vanishes as $T_{q,s}\rightarrow0$ for $q\in(3/5,1)$.

\begin{figure}[H]
    \centering
    \includegraphics[width=\textwidth]{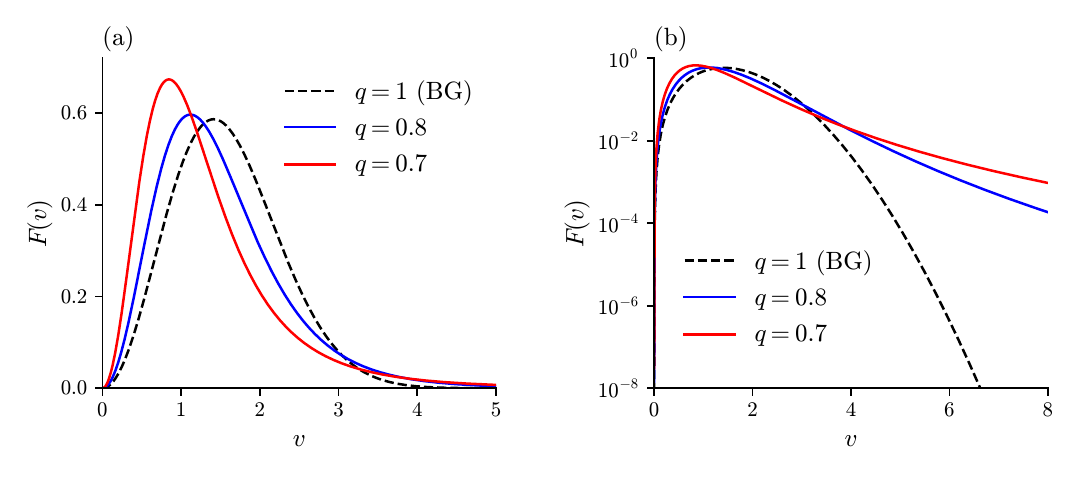}
    \caption{Distribution of speeds $F(v)=4\pi v^{2}\rho(\mathbf{v})$ given by (\ref{eqMBGen}) for $q=0.7,0.8$ and by (\ref{eqMB}) for the Boltzmann-Gibbs (BG) case ($q=1$), with $m=1$ and the same mean energy $U=3/2$ (i.e., $kT=1$ in the BG case), in linear (a) and logarithmic (b) scales.}
    \label{figMB}
\end{figure}

\subsection{The harmonic oscillator}\label{secGho}

The energy spectrum of the harmonic oscillator\footnote{This is actually a shifted version of the energy spectrum, so that $e_{1}=0$.} is $e_{n}=(n-1)\hbar\omega$, $\hbar\omega>0$, $n\geq1$, and $\Omega(e_{n})=1$, $n\geq1$. Let $q\in(0,1)$ and $U>0$. I consider (\ref{maxProblem}) with only the first $N>L(U)$ energy values and use the notation of Theorem \ref{theoTrunc}. By Theorem \ref{theoTrunc}(ii), $\{\beta^{N}_{q}\}_{N\geq N_{1}}$ is strictly increasing with $0<\lim_{N\rightarrow\infty}\beta^{N}_{q}=\beta_{q}\leq1/(1-q)U$. Whether this upper bound is attained is decided by Theorem \ref{theoTrunc}(iv), with
\begin{eqnarray*}
    \Psi(U)=\sum^{\infty}_{n=L(U)+1}\biggl(\frac{(n-1)\hbar\omega}{U}\biggr)^{\frac{1}{q-1}}((n-1)\hbar\omega-U).
\end{eqnarray*}
Up to a positive constant, the general term of this series behaves as $n^{-\frac{q}{1-q}}$ for $n\gg1$, so that $\Psi(U)=+\infty$ if, and only if, $q/(1-q)\leq1$, i.e., $q\leq1/2$.

If $q\leq1/2$, then
\begin{eqnarray*}
    \beta_{q}=\frac{1}{(1-q)U},
\end{eqnarray*}
and, by Theorem \ref{theoTrunc}(iii), all the probability concentrates on the ground level (i.e., $\lim_{N\rightarrow\infty}{\tilde{P}^{N}_{1}=1}$, the mean energy being carried by a vanishing probability at increasingly high levels) and $\lim_{N\rightarrow\infty}T^{N}_{q}=0$ for every $U>0$.

If $q\in(1/2,1)$, then $\delta=1/(1-q)U-\beta_{q}>0$ and, by Theorem \ref{theoTrunc}(v), the optimal distributions of the truncated systems converge to $\tilde{p}_{n}=\tilde{p}_{1}r_{n}$, $n\geq1$, with
\begin{eqnarray}\label{eqHOrn}
    r_{n}=\lim_{N\rightarrow\infty}\frac{\tilde{p}^{N}_{n}}{\tilde{p}^{N}_{1}}=\biggl(\frac{U\delta}{U\delta+\beta_{q}(n-1)\hbar\omega}\biggr)^{\frac{1}{1-q}},
\end{eqnarray}
and the limiting distribution satisfies the normalization and energy constraints, i.e.,
\begin{eqnarray}\label{eqHOConstraints}
    \tilde{p}_{1}=\biggl(\sum^{\infty}_{n=1}r_{n}\biggr)^{-1}\qquad\text{and}\qquad\frac{\sum^{\infty}_{n=1}r_{n}(n-1)\hbar\omega}{\sum^{\infty}_{n=1}r_{n}}=U.
\end{eqnarray}
Notice that $r_{n}\approx(U\delta/\beta_{q}\hbar\omega)^{\frac{1}{1-q}}n^{-\frac{1}{1-q}}$ and $r_{n}(n-1)\approx(U\delta/\beta_{q}\hbar\omega)^{\frac{1}{1-q}}n^{-\frac{q}{1-q}}$ for $n\gg1$, so that both series in (\ref{eqHOConstraints}) converge precisely because $q>1/2$ (i.e., $1<q/(1-q)<1/(1-q)$). In other words, $\beta_{q}$ is the unique solution in $(0,1/(1-q)U)$ of $\sum^{\infty}_{n=1}\exp_{2-q}(-\beta_{q}E_{n})E_{n}=0$, the counterpart of (\ref{eqGenBeta}) for the infinite spectrum.

Moreover, for $q\in(1/2,1)$, \emph{no rescaling is needed}: by Theorem \ref{theoTrunc}(v), the generalized temperature (\ref{eqTemp}) of the truncated system has a finite, positive limit,
\begin{eqnarray}\label{eqHOTemp}
    T_{q}=\lim_{N\rightarrow\infty}T^{N}_{q}=\frac{(1-q)U\delta\,\tilde{p}_{1}^{1-q}}{kq\beta_{q}}=\frac{\tilde{p}_{1}^{1-q}(1-(1-q)\beta_{q}U)}{kq\beta_{q}},
\end{eqnarray}
with $\beta_{q}$ and $\tilde{p}_{1}$ determined by (\ref{eqHOrn}) and (\ref{eqHOConstraints}). As $q\rightarrow1$, $\tilde{p}_{1}^{1-q}\rightarrow1$ and $\beta_{q}\rightarrow\beta_{BG}=\ln((\hbar\omega+U)/U)/\hbar\omega$, so that (\ref{eqHOTemp}) recovers the Boltzmann-Gibbs temperature $T=\hbar\omega/k\ln(1+\hbar\omega/U)$. 

\subsection{The one-dimensional box}\label{secGbox}

The energy spectrum of the particle in a one-dimensional box\footnote{As in Subsection \ref{secGho}, the energy spectrum is shifted so that $e_{1}=0$.} is $e_{n}={(n^{2}-1)}\hbar^{2}\pi^{2}/2m\ell^{2}={(n^{2}-1)}\gamma$, $\gamma>0$, $n\geq1$, with $m$ the mass of the particle and $\ell$ the length of the box, and $\Omega(e_{n})=1$, $n\geq1$. Let $q\in(0,1)$ and $U>0$. I consider (\ref{maxProblem}) with only the first $N>L(U)$ energy values. By Theorem \ref{theoTrunc}(ii), $\{\beta^{N}_{q}\}_{N\geq N_{1}}$ is strictly increasing with $0<\lim_{N\rightarrow\infty}\beta^{N}_{q}=\beta_{q}\leq1/(1-q)U$. In this case, the series of Theorem \ref{theoTrunc}(iv) is
\begin{eqnarray*}
    \Psi(U)=\sum^{\infty}_{n=L(U)+1}\biggl(\frac{(n^{2}-1)\gamma}{U}\biggr)^{\frac{1}{q-1}}((n^{2}-1)\gamma-U),
\end{eqnarray*}
whose general term behaves, up to a positive constant, as $n^{-\frac{2q}{1-q}}$ for $n\gg1$, so that $\Psi(U)=+\infty$ if, and only if, $2q/(1-q)\leq1$, i.e., $q\leq1/3$.

If $q\leq1/3$, then, as in Subsection \ref{secGho},
\begin{eqnarray*}
    \beta_{q}=\frac{1}{(1-q)U},
\end{eqnarray*}
and, by Theorem \ref{theoTrunc}(iii), the limiting distribution is degenerate as well (i.e., $\lim_{N\rightarrow\infty}\tilde{P}^{N}_{1}=1$ and $\lim_{N\rightarrow\infty}T^{N}_{q}=0$).

If $q\in(1/3,1)$, then $\delta=1/(1-q)U-\beta_{q}>0$ and, by Theorem \ref{theoTrunc}(v), the optimal distributions of the truncated systems converge to $\tilde{p}_{n}=\tilde{p}_{1}r_{n}$, $n\geq1$, with
\begin{eqnarray}\label{eqBoxrn}
    r_{n}=\biggl(\frac{U\delta}{U\delta+\beta_{q}(n^{2}-1)\gamma}\biggr)^{\frac{1}{1-q}},\qquad \tilde{p}_{1}=\biggl(\sum^{\infty}_{n=1}r_{n}\biggr)^{-1},\qquad\frac{\sum^{\infty}_{n=1}r_{n}(n^{2}-1)\gamma}{\sum^{\infty}_{n=1}r_{n}}=U.
\end{eqnarray}
Since $r_{n}\approx(U\delta/\beta_{q}\gamma)^{\frac{1}{1-q}}n^{-\frac{2}{1-q}}$ and $r_{n}(n^{2}-1)\approx(U\delta/\beta_{q}\gamma)^{\frac{1}{1-q}}n^{-\frac{2q}{1-q}}$ for $n\gg1$, both series converge precisely because $q>1/3$ (i.e., $1<2q/(1-q)<2/(1-q)$). Once more, no rescaling is needed, and the generalized temperature of the truncated system has the finite, positive limit
\begin{eqnarray}\label{eqBoxTemp}
    T_{q}=\lim_{N\rightarrow\infty}T^{N}_{q}=\frac{\tilde{p}_{1}^{1-q}(1-(1-q)\beta_{q}U)}{kq\beta_{q}},
\end{eqnarray}
with $\beta_{q}$ and $\tilde{p}_{1}$ determined by (\ref{eqBoxrn}). As $q\rightarrow1$, (\ref{eqBoxTemp}) recovers the Boltzmann-Gibbs temperature $T=1/k\beta_{BG}$, with $\beta_{BG}$ given by
\begin{eqnarray*}
    \frac{\sum^{\infty}_{n=1}\exp(-\beta_{BG}(n^{2}-1)\gamma)}{\sum^{\infty}_{n=1}n^{2}\exp(-\beta_{BG}(n^{2}-1)\gamma)}=\frac{\gamma}{U+\gamma}.
\end{eqnarray*}

\section{Concluding remarks}\label{sec7}

This paper derives generalizations of the Boltzmann-Gibbs distribution for systems whose description requires a limiting process over the energy spectrum, by optimizing the nonadditive entropic functional $S_{q}$, $q\in(0,1)$, under the same \emph{linear} mean-energy constraint as in the Boltzmann-Gibbs case, which is the best-behaved case identified in \cite{DogniniTsallis2026}.

The main result concerns the hydrogen atom in free space. For every $q\in(0,1)$, the limiting optimal distribution is a two-tier one, concentrated on the ground level and on the ionization threshold with weights $1-U/e_{ion}$ and $U/e_{ion}$, so that $q$ does not affect the distribution and only impacts the thermodynamic behavior. With the scale factor $s(W)=\ln_{q}W$, the temperature $T_{q,s}=T_{ion}(U/e_{ion})^{1-q}$, $T_{ion}=e_{ion}/kq$, is finite and smooth on $(0,e_{ion})$, and the specific heat vanishes as $T_{q,s}\rightarrow0$, in accordance with the Third Law of Thermodynamics. Therefore, the Electronic Partition Function Paradox is circumvented without truncating the spectrum on external grounds. Furthermore, the remaining systems show how the same framework adapts to different spectra and is capable of furnishing well-behaved thermodynamic predictions. 

Nonetheless, several questions remain open. For instance, the convergence assumptions underlying the continuous limits of Subsections \ref{secGunif} and \ref{secGmb} are not formally proven here. Likewise, the recovery of the Boltzmann-Gibbs results as $q\rightarrow1$ for the infinite spectra of Subsections \ref{secGho} and \ref{secGbox} involves interchanging this limit with $N\rightarrow\infty$, which is also left unproven. Perhaps more fundamentally, the choice of $q\in(0,1)$ that describes a specific system requires further physical considerations or empirical measurements (such as the comparison of $T_{ion}$ with hydrogen ionization data), which extend beyond the scope of this paper.

\section*{Acknowledgments}
Fruitful discussions with C. Tsallis at Centro Brasileiro de Pesquisas Físicas (CBPF) are gratefully acknowledged.

\appendix
\phantomsection
\section*{Appendix}\label{app}

\begin{proof}[Proof of Corollary~{\upshape\ref{corSum}}]
Let $x_{i}=\exp_{2-q}(-\beta_{q}E_{i})$, $1\leq i\leq W$. By Theorem \ref{theoDT}, $\tilde{p}\in\mathbb{R}^{W}_{++}$, so that $1+(1-q)\beta_{q}E_{i}>0$ and $x_{i}=(1+(1-q)\beta_{q}E_{i})^{-\frac{1}{1-q}}\in(0,+\infty)$. Then, $x^{q-1}_{i}=1+(1-q)\beta_{q}E_{i}$, and therefore $x^{q}_{i}=x_{i}+(1-q)\beta_{q}x_{i}E_{i}$. Summing over $i$ and using (\ref{eqGenBeta}), we obtain
\begin{eqnarray*}
    g_{q}(\beta_{q},U)=\sum^{W}_{i=1}x^{q}_{i}=\sum^{W}_{i=1}x_{i}+(1-q)\beta_{q}\sum^{W}_{i=1}x_{i}E_{i}=\sum^{W}_{i=1}x_{i},
\end{eqnarray*}
which is (\ref{eqSumPq}). Thus, by (\ref{eqGenDistribution}), $\tilde{p}_{i}=x_{i}/g_{q}(\beta_{q},U)$, $1\leq i\leq W$, and
\begin{eqnarray*}
    \sum^{W}_{i=1}\tilde{p}^{q}_{i}=\frac{\sum^{W}_{i=1}x^{q}_{i}}{g_{q}(\beta_{q},U)^{q}}=\frac{g_{q}(\beta_{q},U)}{g_{q}(\beta_{q},U)^{q}}=g_{q}(\beta_{q},U)^{1-q},
\end{eqnarray*}
as stated.
\end{proof}

\begin{proof}[Proof of Theorem~{\upshape\ref{theoTrunc}}]
Let $L\equiv L(U)$, $N>L$, and $\phi:(-\infty,1/(1-q)U)\rightarrow\mathbb{R}_{++}$ and $\psi^{N}:(-1/(1-q)E_{N},+\infty)\rightarrow\mathbb{R}_{++}$ be given by
\begin{eqnarray*}
    \phi(x)&=&-\sum^{L}_{n=1}\Omega(e_{n})\exp_{2-q}(-xE_{n})E_{n}=\sum^{L}_{n=1}\Omega(e_{n})\exp_{2-q}(x\vert E_{n}\vert)\vert E_{n}\vert,\\
    \psi^{N}(x)&=&\sum^{N}_{n=L+1}\Omega(e_{n})\exp_{2-q}(-xE_{n})E_{n}=\sum^{N}_{n=L+1}\Omega(e_{n})\exp_{2-q}(-x\vert E_{n}\vert)\vert E_{n}\vert,
\end{eqnarray*}
with the last equalities due to $E_{n}\leq0$, $n\leq L$, and $E_{n}>0$, $n>L$. Since $\vert E_{1}\vert=U$ is the largest value among $\{\vert E_{n}\vert\}_{n\leq L}$, $\phi(\cdot)$ is continuous and strictly increasing, with $\lim_{x\rightarrow-\infty}\phi(x)=0$ and $\lim_{x\rightarrow1/(1-q)U}\phi(x)=+\infty$. Also, $\psi^{N}(\cdot)$ is continuous and strictly decreasing, with $\psi^{N+1}(x)>\psi^{N}(x)$, $\lim_{x\rightarrow+\infty}\psi^{N}(x)=0$ and $\lim_{x\rightarrow-1/(1-q)E_{N}}\psi^{N}(x)=+\infty$. 

Notice that (\ref{eqBetaTrunc}) is equivalent to $\phi(\beta^{N}_{q})=\psi^{N}(\beta^{N}_{q})$, so that the Intermediate Value Theorem and the strict monotonicity of $\phi(\cdot)$ and $\psi^{N}(\cdot)$ imply the existence of a unique $\beta^{N}_{q}<1/{(1-q)}U$ solving (\ref{eqBetaTrunc}). Furthermore, for $\delta^{N}=1/{(1-q)}U-\beta^{N}_{q}>0$, we have $\exp_{2-q}(-\beta^{N}_{q}E_{n})=({(1-q)}{(U\delta^{N}+\beta^{N}_{q}e_{n})})^{\frac{1}{q-1}}$, $1\leq n\leq N$, and (\ref{eqGenDistribution}) implies (\ref{eqRatios}). This proves (i).

Next, notice that $\psi^{N}(0)-\phi(0)=\sum^{N}_{n=1}\Omega(e_{n})E_{n}$. Since $\Omega(e_{n})\geq1$ and $E_{n}\geq E_{L+1}>0$ for $n>L$, this sum is strictly increasing in $N>L$ and bounded from below by $\sum^{L}_{n=1}\Omega(e_{n})E_{n}+(N-L)E_{L+1}$, so that there is $N_{1}>L$ such that $\sum^{N}_{n=1}\Omega(e_{n})E_{n}>0$ for $N\geq N_{1}$. For such $N$, we have $\psi^{N}(0)>\phi(0)$ and, therefore, $\beta^{N}_{q}>0$. Suppose, then, $N\geq N_{1}$ and $0<\beta^{N+1}_{q}\leq\beta^{N}_{q}$. Then,
\begin{eqnarray*}
    \phi(\beta^{N+1}_{q})\leq \phi(\beta^{N}_{q})=\psi^{N}(\beta^{N}_{q})<\psi^{N+1}(\beta^{N}_{q})\leq \psi^{N+1}(\beta^{N+1}_{q})=\phi(\beta^{N+1}_{q}),
\end{eqnarray*}
a contradiction. We conclude that $\{\beta^{N}_{q}\}_{N\geq N_{1}}$ is strictly increasing and, since it is bounded by $1/(1-q)U$, it converges to some $\beta_{q}\in(0,1/(1-q)U]$. This proves (ii).

Next, suppose $\lim_{N\rightarrow\infty}\delta^{N}=0$. By (\ref{eqRatios}), for $n\geq2$, we have
\begin{eqnarray*}
\tilde{P}^{N}_{n}=\tilde{P}^{N}_{1}\frac{\Omega(e_{n})}{\Omega(e_{1})}\frac{\tilde{p}^{N}_{n}}{\tilde{p}^{N}_{1}}\leq \frac{\Omega(e_{n})}{\Omega(e_{1})}\biggl(\frac{U\delta^{N}}{U\delta^{N}+\beta^{N}_{q}e_{n}}\biggr)^{\frac{1}{1-q}}\rightarrow0,
\end{eqnarray*}
since $\beta^{N}_{q}e_{n}\geq\beta^{N_{1}}_{q}e_{2}>0$. Also, by (\ref{eqSumPq}),
\begin{eqnarray*}
    g_{q}(\beta^{N}_{q},U)=\sum^{N}_{n=1}\Omega(e_{n})\exp_{2-q}(-\beta^{N}_{q}E_{n})\geq \Omega(e_{1})\exp_{2-q}(\beta^{N}_{q}U)=\frac{\Omega(e_{1})}{((1-q)U\delta^{N})^{\frac{1}{1-q}}}\rightarrow+\infty,
\end{eqnarray*}
and (\ref{eqTemp}) implies $\lim_{N\rightarrow\infty}T^{N}_{q}=0$, since $\beta^{N}_{q}\rightarrow1/(1-q)U>0$. Finally, let $e_{\infty}=\sup\{e_{n}\}_{n\geq1}\in(U,+\infty]$. For $2\leq M\leq N$, since $e_{1}=0$ and $e_{n}\geq e_{M}$ for $n\geq M$, the normalization and energy constraints of the truncated system (i.e., $\sum^{N}_{n=1}\tilde{P}^{N}_{n}=1$ and $\sum^{N}_{n=1}\tilde{P}^{N}_{n}e_{n}=U$) imply
\begin{eqnarray*}
    e_{M}\biggl(1-\tilde{P}^{N}_{1}-\sum^{M-1}_{n=2}\tilde{P}^{N}_{n}\biggr)=e_{M}\sum^{N}_{n=M}\tilde{P}^{N}_{n}\leq\sum^{N}_{n=M}\tilde{P}^{N}_{n}e_{n}=U-\sum^{M-1}_{n=2}\tilde{P}^{N}_{n}e_{n},
\end{eqnarray*}
so that, since $\lim_{N\rightarrow\infty}\tilde{P}^{N}_{n}=0$ for $2\leq n<M$, taking the limit as $N\rightarrow\infty$ and then as $M\rightarrow\infty$ yields $\limsup_{N\rightarrow\infty}(1-\tilde{P}^{N}_{1})\leq U/e_{\infty}$ (with $U/e_{\infty}=0$ if $e_{\infty}=+\infty$). Since the energy constraint also implies $U\leq e_{\infty}(1-\tilde{P}^{N}_{1})$ if $e_{\infty}<+\infty$, we conclude that $\lim_{N\rightarrow\infty}\tilde{P}^{N}_{1}=1-U/e_{\infty}$. This proves (iii).

Next, notice that, since $1+E_{n}/U=e_{n}/U$,
\begin{eqnarray*}
    \psi^{N}\biggl(\frac{1}{(1-q)U}\biggr)=\sum^{N}_{n=L+1}\Omega(e_{n})\biggl(\frac{e_{n}}{U}\biggr)^{\frac{1}{q-1}}(e_{n}-U)\rightarrow\Psi(U)
\end{eqnarray*}
as $N\rightarrow\infty$. Suppose $\Psi(U)=+\infty$. Since $\psi^{N}(\cdot)$ is decreasing and $\beta^{N}_{q}<1/(1-q)U$, we have $\phi(\beta^{N}_{q})=\psi^{N}(\beta^{N}_{q})\geq\psi^{N}(1/(1-q)U)\rightarrow+\infty$ and, since $\phi(\cdot)$ is bounded on $(-\infty,x]$ for every $x<1/(1-q)U$, we conclude that $\beta_{q}=1/(1-q)U$. Conversely, suppose $\Psi(U)<+\infty$. For $E>0$ and $0<x\leq x'$, we have $1+(1-q)x'E\leq(x'/x)(1+(1-q)xE)$, so that
\begin{eqnarray}\label{eqExpCompare}
    \exp_{2-q}(-xE)\leq\biggl(\frac{x'}{x}\biggr)^{\frac{1}{1-q}}\exp_{2-q}(-x'E).
\end{eqnarray}
Since $\{\beta^{N}_{q}\}_{N\geq N_{1}}$ is increasing and $\psi^{N}(\cdot)$ is decreasing, (\ref{eqExpCompare}) with $x=\beta^{N_{1}}_{q}$ and $x'=1/(1-q)U$ yields, for $N\geq N_{1}$,
\begin{eqnarray}\label{eqPhiBound}
    \phi(\beta^{N}_{q})=\psi^{N}(\beta^{N}_{q})\leq\psi^{N}(\beta^{N_{1}}_{q})\leq((1-q)U\beta^{N_{1}}_{q})^{\frac{1}{q-1}}\,\Psi(U)<+\infty.
\end{eqnarray}
Since $\lim_{x\rightarrow1/(1-q)U}\phi(x)=+\infty$, (\ref{eqPhiBound}) implies that $\{\beta^{N}_{q}\}_{N\geq N_{1}}$ is bounded away from $1/(1-q)U$, i.e., $\beta_{q}<1/(1-q)U$. This proves (iv).

Finally, suppose $\delta=1/(1-q)U-\beta_{q}>0$, so that $\Psi(U)<+\infty$ by (iv). For $N\geq N_{1}$, let $r^{N}_{n}=\tilde{p}^{N}_{n}/\tilde{p}^{N}_{1}$, $1\leq n\leq N$, and $r^{N}_{n}=0$, $n>N$, so that, by (\ref{eqRatios}), $\lim_{N\rightarrow\infty}r^{N}_{n}=r_{n}$, $n\geq1$. Since $\{\delta^{N}\}_{N\geq N_{1}}$ is decreasing and $\{\beta^{N}_{q}\}_{N\geq N_{1}}$ is increasing, (\ref{eqRatios}) also implies
\begin{eqnarray*}
    r^{N}_{n}\leq\bar{r}_{n}\equiv\biggl(\frac{U\delta^{N_{1}}}{U\delta^{N_{1}}+\beta^{N_{1}}_{q}e_{n}}\biggr)^{\frac{1}{1-q}}=\frac{\exp_{2-q}(-\beta^{N_{1}}_{q}E_{n})}{\exp_{2-q}(\beta^{N_{1}}_{q}U)},
\end{eqnarray*}
for $n\geq1$, and, by (\ref{eqPhiBound}), $\sum^{\infty}_{n=L+1}\Omega(e_{n})\bar{r}_{n}E_{n}=\lim_{N\rightarrow\infty}\psi^{N}(\beta^{N_{1}}_{q})/\exp_{2-q}(\beta^{N_{1}}_{q}U)<+\infty$. Since $E_{n}\geq E_{L+1}>0$ and $e_{n}\leq(1+U/E_{L+1})E_{n}$ for $n>L$, both $\sum^{\infty}_{n=1}\Omega(e_{n})\bar{r}_{n}$ and $\sum^{\infty}_{n=1}\Omega(e_{n})\bar{r}_{n}e_{n}$ converge, and the Dominated Convergence Theorem yields
\begin{eqnarray*}
    \lim_{N\rightarrow\infty}\sum^{N}_{n=1}\Omega(e_{n})r^{N}_{n}=\sum^{\infty}_{n=1}\Omega(e_{n})r_{n}\qquad\text{and}\qquad\lim_{N\rightarrow\infty}\sum^{N}_{n=1}\Omega(e_{n})r^{N}_{n}e_{n}=\sum^{\infty}_{n=1}\Omega(e_{n})r_{n}e_{n}.
\end{eqnarray*}
The normalization and energy constraints of the truncated system read $\tilde{p}^{N}_{1}\sum^{N}_{n=1}\Omega(e_{n})r^{N}_{n}=1$ and $\tilde{p}^{N}_{1}\sum^{N}_{n=1}\Omega(e_{n})r^{N}_{n}e_{n}=U$. Therefore, $\lim_{N\rightarrow\infty}\tilde{p}^{N}_{1}=\tilde{p}_{1}>0$, $\lim_{N\rightarrow\infty}\tilde{p}^{N}_{n}=\tilde{p}_{1}r_{n}=\tilde{p}_{n}$, $n\geq1$, $\sum^{\infty}_{n=1}\Omega(e_{n})\tilde{p}_{n}=1$ and $\sum^{\infty}_{n=1}\Omega(e_{n})\tilde{p}_{n}e_{n}=U$. Since $r_{n}=\exp_{2-q}(-\beta_{q} E_{n})/\exp_{2-q}(\beta_{q} U)$, these two equalities imply $\sum^{\infty}_{n=1}\Omega(e_{n})\exp_{2-q}(-\beta_{q} E_{n})E_{n}=0$. The solution of this equation in $(0,1/(1-q)U)$ is unique because, by (\ref{eqExpCompare}) and $\Psi(U)<+\infty$, its left-hand side is the difference between $\lim_{N\rightarrow\infty}\psi^{N}(\cdot)$, which is finite and strictly decreasing in this interval, and $\phi(\cdot)$, which is strictly increasing. Finally, $\tilde{P}^{N}_{1}/\Omega(e_{1})=\tilde{p}^{N}_{1}$, $1-(1-q)\beta^{N}_{q}U=(1-q)U\delta^{N}$ and (\ref{eqTempTrunc}) imply (\ref{eqTempLim}). This proves (v).
\end{proof}

\begin{proof}[Proof of Proposition~{\upshape\ref{propParadox}}]
Let $L\equiv L(U)$, $N>L$, and define $\phi(\cdot)$ and $\psi^{N}(\cdot)$ as in the proof of Theorem \ref{theoTrunc}, with $\exp$ in place of $\exp_{2-q}$ (now over $\mathbb{R}$), so that the solution $\beta^{N}_{BG}$ of (\ref{eqBGBeta}) for the truncated spectrum satisfies $\phi(\beta^{N}_{BG})=\psi^{N}(\beta^{N}_{BG})$. Then, $\phi(\cdot)$ is strictly increasing, $\psi^{N}(\cdot)$ is strictly decreasing and $\psi^{N+1}(x)>\psi^{N}(x)$, $x\in\mathbb{R}$. As in the proof of Theorem \ref{theoTrunc}(ii), there is $N_{1}>L$ such that $\sum^{N}_{n=1}\Omega(e_{n})E_{n}>0$ for $N\geq N_{1}$, and the same reasoning shows that $\{\beta^{N}_{BG}\}_{N\geq N_{1}}$ is positive and strictly increasing. Suppose $\lim_{N\rightarrow\infty}\beta^{N}_{BG}=\hat{\beta}\in\mathbb{R}$. Then, since $0<E_{n}<e_{ion}-U$ for $n>L$, we have
\begin{eqnarray*}
    \phi(\hat{\beta})&=&\lim_{N\rightarrow\infty}\phi(\beta^{N}_{BG})=\lim_{N\rightarrow\infty}\psi^{N}(\beta^{N}_{BG})\geq\lim_{N\rightarrow\infty}\psi^{N}(\hat{\beta})\\
    &\geq&\exp(-\hat{\beta}(e_{ion}-U))E_{L+1}\lim_{N\rightarrow\infty}\sum^{N}_{n=L+1}2n^{2}=+\infty,
\end{eqnarray*}
a contradiction. We conclude that $\lim_{N\rightarrow\infty}\beta^{N}_{BG}=+\infty$ and, therefore, $\lim_{N\rightarrow\infty}T^{N}_{BG}=0$.

Next, by (\ref{eqBGDistribution}),
\begin{eqnarray*}
    \tilde{P}^{N}_{n}=n^{2}\tilde{P}^{N}_{1}\exp(-\beta^{N}_{BG}e_{n})\leq n^{2}\exp(-\beta^{N}_{BG}e_{n}),
\end{eqnarray*}
for $1\leq n\leq N$, so that $\lim_{N\rightarrow\infty}\tilde{P}^{N}_{n}=0$ for $n\geq2$. Then, the last part of the proof of Theorem \ref{theoTrunc}(iii), which only uses this fact and the normalization and energy constraints, yields $\lim_{N\rightarrow\infty}\tilde{P}^{N}_{1}=1-U/e_{ion}$.

Finally, let $B_{N}=\beta^{N}_{BG}e_{ion}$ and $A_{N}=\sum^{N}_{n=2}2n^{2}\exp(B_{N}/n^{2})$. Since $e_{n}=e_{ion}-e_{ion}/n^{2}$, the normalization constraint yields
\begin{eqnarray*}
    \frac{1-\tilde{P}^{N}_{1}}{\tilde{P}^{N}_{1}}=\sum^{N}_{n=2}\frac{\tilde{P}^{N}_{n}}{\tilde{P}^{N}_{1}}=\sum^{N}_{n=2}n^{2}\exp(-\beta^{N}_{BG}e_{n})=\frac{A_{N}}{2\exp(B_{N})},
\end{eqnarray*}
i.e., $\exp(B_{N})=\kappa_{N}A_{N}$, where $\kappa_{N}=\tilde{P}^{N}_{1}/2(1-\tilde{P}^{N}_{1})$ converges to $(e_{ion}-U)/2U$. Since $A_{N}\leq\exp(B_{N}/4)W(N)$, we have $\exp(B_{N}/4)\leq(\kappa_{N}W(N))^{1/3}$ and, in particular, $B_{N}/(\ln N)^{2}\rightarrow0$. Let $M_{N}$ be the smallest integer greater than $\ln N$. Splitting the sum $A_{N}$ at $n=M_{N}$, we obtain
\begin{eqnarray*}
    W(N)-2\leq A_{N}&\leq&\exp(B_{N}/4)W(M_{N})+\exp(B_{N}/M^{2}_{N})W(N)\\
    &\leq&(\kappa_{N}W(N))^{\frac{1}{3}}W(M_{N})+\exp(B_{N}/M^{2}_{N})W(N),
\end{eqnarray*}
so that, since $W(N)=2N^{3}/3+o(N^{3})$, we have $\lim_{N\rightarrow\infty}A_{N}/W(N)=1$. Therefore,
\begin{eqnarray*}
    \lim_{N\rightarrow\infty}\frac{\exp(B_{N})}{W(N)}=\frac{e_{ion}-U}{2U},
\end{eqnarray*}
which is the first limit in (\ref{eqBGRescaled}). The second one follows from $kT^{N}_{BG}=e_{ion}/B_{N}$ and from the divergence of $\ln W(N)$.
\end{proof}

\begin{proof}[Proof of Theorem~{\upshape\ref{theoH}}]
Let $L\equiv L(U)$, and let $\phi(\cdot)$ and $\psi^{N}(\cdot)$ be as in the proof of Theorem \ref{theoTrunc}. By Theorem \ref{theoTrunc}(ii), $\{\beta^{N}_{q}\}_{N\geq N_{1}}$ is positive and strictly increasing and, since $\Psi(U)=+\infty$ (see Subsection \ref{secGh}), Theorem \ref{theoTrunc}(iv) yields $\lim_{N\rightarrow\infty}\beta^{N}_{q}=\beta_{q}=1/(1-q)U$, i.e., $\{\delta^{N}\}_{N\geq N_{1}}$ is strictly decreasing with $\lim_{N\rightarrow\infty}\delta^{N}=0$. Next, notice that $\exp_{2-q}(\beta^{N}_{q}U)=((1-q)\delta^{N}U)^{\frac{1}{q-1}}$ and
\begin{eqnarray}\label{eqParadoxProof1}
    \phi(\beta^{N}_{q})&=&2((1-q)\delta^{N}U)^{\frac{1}{q-1}}U+\sum^{L}_{n=2}2n^{2}(1-(1-q)\beta^{N}_{q}\vert E_{n}\vert)^{\frac{1}{q-1}}\vert E_{n}\vert\nonumber\\
    &=&2(1-q)^{\frac{1}{q-1}}(\delta^{N})^{\frac{1}{q-1}}U^{\frac{q}{q-1}}\biggl(1+(\delta^{N})^{\frac{1}{1-q}}\sum^{L}_{n=2}\frac{n^{2}\vert E_{n}\vert(1-q)^{\frac{1}{1-q}}U^{\frac{q}{1-q}}}{(1-(1-q)\beta^{N}_{q}\vert E_{n}\vert)^{\frac{1}{1-q}}}\biggr).
\end{eqnarray}
Since $1-(1-q)\beta^{N}_{q}\vert E_{n}\vert\rightarrow e_{n}/U>0$, $2\leq n\leq L$, and $\delta^{N}\rightarrow0$, the sum multiplied by $(\delta^{N})^{\frac{1}{1-q}}$ in (\ref{eqParadoxProof1}) vanishes as $N\rightarrow\infty$. Also,
\begin{eqnarray}\label{eqGNSums}
    \psi^{N}(\beta^{N}_{q})&=&\sum^{N}_{n=L+1}2n^{2}(1+(1-q)\beta^{N}_{q}E_{n})^{\frac{1}{q-1}}E_{n}\nonumber\\
    &=&(1-q)^{\frac{1}{q-1}}\sum^{N}_{n=L+1}2n^{2}(U\delta^{N}+\beta^{N}_{q}e_{n})^{\frac{1}{q-1}}\biggl(e_{ion}-U-\frac{e_{ion}}{n^{2}}\biggr)\nonumber\\
    &=&(1-q)^{\frac{1}{q-1}}(e_{ion}-U)\sum^{N}_{n=L+1}2n^{2}(U\delta^{N}+\beta^{N}_{q}e_{n})^{\frac{1}{q-1}}\nonumber\\
    &&-(1-q)^{\frac{1}{q-1}}e_{ion}\sum^{N}_{n=L+1}2(U\delta^{N}+\beta^{N}_{q}e_{n})^{\frac{1}{q-1}}.
\end{eqnarray}
For $N\geq N_{1}$, we have $U\delta^{N}+\beta^{N}_{q}e_{n}\geq\beta^{N_{1}}_{q}e_{L+1}>0$, $n>L$, and, therefore,
\begin{eqnarray}\label{eqLimG1}
    0\leq\lim_{N\rightarrow\infty}\frac{\sum^{N}_{n=L+1}2(U\delta^{N}+\beta^{N}_{q}e_{n})^{\frac{1}{q-1}}}{N^{3}}\leq\lim_{N\rightarrow\infty}\frac{2N(\beta^{N_{1}}_{q}e_{L+1})^{\frac{1}{q-1}}}{N^{3}}=0.
\end{eqnarray}
Next, I write
\begin{eqnarray*}
    \sum^{N}_{n=L+1}2n^{2}(U\delta^{N}+\beta^{N}_{q}e_{n})^{\frac{1}{q-1}}=\sum^{N}_{n=L+1}2n^{2}(\beta_{q} e_{ion}+\sigma(n,N))^{\frac{1}{q-1}},\\
    \sigma(n,N)=U\delta^{N}+(\beta^{N}_{q}-\beta_{q})e_{ion}-\frac{\beta^{N}_{q}e_{ion}}{n^{2}},
\end{eqnarray*}
for $L+1\leq n\leq N$. For $0<\varepsilon<\beta_{q} e_{ion}/2$, let $M_{1}\geq L+1$ be such that $N\geq n\geq M_{1}$ implies $\vert\sigma(n,N)\vert<\varepsilon$. Then, for $N\geq M_{1}$, the Mean Value Theorem yields
\begin{eqnarray*}
    \biggl\vert\sum^{N}_{n=M_{1}}2n^{2}\Bigl((\beta_{q} e_{ion}+\sigma(n,N))^{\frac{1}{q-1}}-(\beta_{q} e_{ion})^{\frac{1}{q-1}}\Bigr)\biggr\vert
    &\leq&\frac{\varepsilon}{(1-q)(\beta_{q} e_{ion}-\varepsilon)^{\frac{2-q}{1-q}}}\sum^{N}_{n=M_{1}}2n^{2}\\
    &=&\frac{\varepsilon}{(1-q)(\beta_{q} e_{ion}-\varepsilon)^{\frac{2-q}{1-q}}}\biggl(\frac{2N^{3}}{3}+o(N^{3})\biggr).
\end{eqnarray*}
Since the finitely many terms with $L+1\leq n<M_{1}$ are $o(N^{3})$, we have
\begin{eqnarray*}
    \limsup_{N\rightarrow\infty}\frac{\bigl\vert\sum^{N}_{n=L+1}2n^{2}((\beta_{q} e_{ion}+\sigma(n,N))^{\frac{1}{q-1}}-(\beta_{q} e_{ion})^{\frac{1}{q-1}})\bigr\vert}{N^{3}}\leq\frac{2\varepsilon}{3(1-q)(\beta_{q} e_{ion}-\varepsilon)^{\frac{2-q}{1-q}}},
\end{eqnarray*}
and, since $\varepsilon>0$ is arbitrarily small,
\begin{eqnarray}\label{eqLimG2}
    \lim_{N\rightarrow\infty}\frac{\sum^{N}_{n=L+1}2n^{2}(U\delta^{N}+\beta^{N}_{q}e_{n})^{\frac{1}{q-1}}}{N^{3}}=\frac{2(\beta_{q} e_{ion})^{\frac{1}{q-1}}}{3}.
\end{eqnarray}
Therefore, (\ref{eqGNSums}), (\ref{eqLimG1}) and (\ref{eqLimG2}) imply
\begin{eqnarray*}
    \lim_{N\rightarrow\infty}\frac{\psi^{N}(\beta^{N}_{q})}{N^{3}}=\frac{2(e_{ion}-U)}{3}((1-q)\beta_{q} e_{ion})^{\frac{1}{q-1}}=\frac{2(e_{ion}-U)}{3}\biggl(\frac{U}{e_{ion}}\biggr)^{\frac{1}{1-q}}.
\end{eqnarray*}
Since $\phi(\beta^{N}_{q})=\psi^{N}(\beta^{N}_{q})$, (\ref{eqParadoxProof1}) implies
\begin{eqnarray*}
    1&=&\lim_{N\rightarrow\infty}\frac{N^{3}}{\psi^{N}(\beta^{N}_{q})}\frac{\phi(\beta^{N}_{q})}{N^{3}}\\
    &=&\frac{3e_{ion}^{\frac{1}{1-q}}}{2(e_{ion}-U)U^{\frac{1}{1-q}}}\,\frac{2}{(1-q)^{\frac{1}{1-q}}U^{\frac{q}{1-q}}}\lim_{N\rightarrow\infty}\frac{1}{(\delta^{N})^{\frac{1}{1-q}}N^{3}}\\
    &=&\frac{3e_{ion}^{\frac{1}{1-q}}}{(1-q)^{\frac{1}{1-q}}(e_{ion}-U)U^{\frac{1+q}{1-q}}}\biggl(\lim_{N\rightarrow\infty}\frac{1}{\delta^{N}N^{3(1-q)}}\biggr)^{\frac{1}{1-q}},
\end{eqnarray*}
so that
\begin{eqnarray}\label{eqLimDeltaN}
    \lim_{N\rightarrow\infty}N^{3(1-q)}\delta^{N}=\frac{3^{1-q}e_{ion}}{(1-q)(e_{ion}-U)^{1-q}U^{1+q}}.
\end{eqnarray}
Next, since $\lim_{N\rightarrow\infty}\delta^{N}=0$ and $\sup\{e_{n}\}_{n\geq1}=e_{ion}$, Theorem \ref{theoTrunc}(iii) yields $\lim_{N\rightarrow\infty}\tilde{P}^{N}_{1}={1-U/e_{ion}}$ and $\lim_{N\rightarrow\infty}\tilde{P}^{N}_{n}=0$ for $n\geq2$. Moreover, since $g_{q}(\beta^{N}_{q},U)=2\exp_{2-q}(\beta^{N}_{q}U)/\tilde{P}^{N}_{1}$ (see (\ref{eqTempTrunc})), we have
\begin{eqnarray*}
    g_{q}(\beta^{N}_{q},U)^{1-q}=\biggl(\frac{2}{\tilde{P}^{N}_{1}}\biggr)^{1-q}\frac{1}{(1-q)\delta^{N}U},
\end{eqnarray*}
and, since $\ln_{q}W(N)=(W(N)^{1-q}-1)/(1-q)$ with $W(N)=2N^{3}/3+o(N^{3})$, we have
\begin{eqnarray*}
\lim_{N\rightarrow\infty}\frac{\ln_{q}W(N)}{N^{3(1-q)}}=\frac{(2/3)^{1-q}}{1-q}.
\end{eqnarray*}
Therefore, (\ref{eqHTempN}) and (\ref{eqLimDeltaN}) imply
\begin{eqnarray*}
    \lim_{N\rightarrow\infty}T^{N}_{q,s}&=&\frac{1}{kq}\lim_{N\rightarrow\infty}\frac{\ln_{q}W(N)}{N^{3(1-q)}}\,\frac{(1-q)U}{\beta^{N}_{q}}\biggl(\frac{\tilde{P}^{N}_{1}}{2}\biggr)^{1-q}\delta^{N}N^{3(1-q)}\\
    &=&\frac{1}{kq}\,\frac{2^{1-q}}{3^{1-q}(1-q)}\,(1-q)^{2}U^{2}\biggl(\frac{e_{ion}-U}{2e_{ion}}\biggr)^{1-q}\frac{3^{1-q}e_{ion}}{(1-q)(e_{ion}-U)^{1-q}U^{1+q}}\\
    &=&\frac{U^{1-q}e^{q}_{ion}}{kq}.
\end{eqnarray*}
Finally, since $S^{N}_{q}/k=\ln_{q}g_{q}(\beta^{N}_{q},U)=(g_{q}(\beta^{N}_{q},U)^{1-q}-1)/(1-q)$ and both $g_{q}(\beta^{N}_{q},U)^{1-q}$ and $W(N)^{1-q}$ diverge, (\ref{eqLimDeltaN}) implies
\begin{eqnarray*}
    \lim_{N\rightarrow\infty}\frac{S^{N}_{q}}{k\ln_{q}W(N)}&=&\lim_{N\rightarrow\infty}\frac{g_{q}(\beta^{N}_{q},U)^{1-q}}{W(N)^{1-q}}=\lim_{N\rightarrow\infty}\biggl(\frac{2}{\tilde{P}^{N}_{1}}\biggr)^{1-q}\frac{1}{(1-q)U\delta^{N}N^{3(1-q)}}\,\frac{N^{3(1-q)}}{W(N)^{1-q}}\\
    &=&\biggl(\frac{2e_{ion}}{e_{ion}-U}\biggr)^{1-q}\frac{(e_{ion}-U)^{1-q}U^{q}}{3^{1-q}e_{ion}}\biggl(\frac{3}{2}\biggr)^{1-q}=\biggl(\frac{U}{e_{ion}}\biggr)^{q},
\end{eqnarray*}
as stated.
\end{proof}

\begin{proof}[Proof of Corollary~{\upshape\ref{corH}}]
Let $\mathcal{F}^{N}$ and $\mathcal{F}$ be the cumulative distribution functions of $H^{N}$ and $H$, respectively. For $e<0$, $\mathcal{F}^{N}(e)=\mathcal{F}(e)=0$, and, for $e\geq e_{ion}$, $\mathcal{F}^{N}(e)=\mathcal{F}(e)=1$. For $0\leq e<e_{ion}$, there is $n(e)\geq1$ such that $e_{n(e)}\leq e<e_{n(e)+1}$, so that, by Theorem \ref{theoH},
\begin{eqnarray*}
    \lim_{N\rightarrow\infty}\mathcal{F}^{N}(e)=\lim_{N\rightarrow\infty}\sum^{n(e)}_{n=1}\tilde{P}^{N}_{n}=1-\frac{U}{e_{ion}}=\mathcal{F}(e).
\end{eqnarray*}
Therefore, $\mathcal{F}^{N}(e)\rightarrow\mathcal{F}(e)$ for every $e\in\mathbb{R}$ and, in particular, at every continuity point of $\mathcal{F}$.
\end{proof}

\bibliographystyle{unsrtnat}
\bibliography{references}
\end{document}